\documentclass[journal,twoside,web]{ieeecolor}

\pdfoutput=1
\newtheorem{definition}{Definition}
\newtheorem{theorem}{Theorem}
\newtheorem{lemma}{Lemma}
\newtheorem{proposition}{Proposition} 
\newtheorem{corollary}{Corollary}
\newtheorem{remark}{Remark}
\newtheorem{example}{Example}
\newtheorem{algorithm}{Algorithm}

\usepackage{generic}
\usepackage{cite}
\usepackage{amsmath,amssymb,amsfonts}
\usepackage{algorithmic}
\usepackage{graphicx}
\usepackage{textcomp}
\def\BibTeX{{\rm B\kern-.05em{\sc i\kern-.025em b}\kern-.08em
    T\kern-.1667em\lower.7ex\hbox{E}\kern-.125emX}}
\begin{document}
\title{Classifications of Single-input Lower Triangular Forms}
\author{Duan Zhang and Ying Sun
\thanks{This work has been submitted for possible publication. Copyright may be transferred without notice, after which this version may no longer be accessible.}
\thanks{D. Zhang is with the College of Computer Science and Technology, Zhejiang University of Technology, Hangzhou, 310023 China (e-mail: dzhang@zjut.edu.cn). }
\thanks{Y. Sun is with the School of Civil Engineering and Architecture, Zhejiang Sci-Tech University, Hangzhou, 310018, China (e-mail: suny\_kangaroo@zstu.edu.cn). }}

\maketitle

\begin{abstract}

The purposes of this paper are to classify lower triangular forms and to determine under what conditions a nonlinear system is equivalent to a specific type of lower triangular forms. According to the least multi-indices and the greatest essential multi-index sets, which are introduced as new notions and can be obtained from the system equations, two classification schemes of lower triangular forms are constructed. It is verified that the type that a given lower triangular form belongs to is invariant under any lower triangular coordinate transformation. Therefore, although a nonlinear system equivalent to a lower triangular form is also equivalent to many other appropriate lower triangular forms, there is only one type that the system can be transformed into. Each of the two classifications induces a classification of all the systems that are equivalent to lower triangular forms. A new method for transforming a nonlinear system into a lower triangular form, if it is possible, is provided to find what type the system belongs to. Additionally, by using the differential geometric control theory, several necessary and sufficient conditions under which a nonlinear system is locally feedback equivalent to a given type of lower triangular form are established. An example is given to illustrate how to determine which type of lower triangular form a given nonlinear system is equivalent to without performing an equivalent transformation.

\end{abstract}


\section{Introduction}
\label{sec:introduction}

\IEEEPARstart{S}{ince} nonlinear 
phenomena are widely present in nature and many industrial processes, the studies of nonlinear control systems are of obvious practical value \cite{isidori1995nonlinear, sepulchre2012constructive, sastry2013nonlinear}. Lower triangular forms are a class of nonlinear systems attracting considerable attention. For example, backstepping, as a powerful control strategy for lower triangular systems, has been developed based on the cascade structures of these systems \cite{sepulchre2012constructive, vcelikovsky1996equivalence, zhang2012new, ma2010backstepping, su2015cooperative, wang2009robust, fotiadis2020prescribed}. Many exciting results have been obtained for some special classes of lower triangular forms, such as strict feedback forms \cite{zhang2000adaptive, tang2003adaptive, bechlioulis2009adaptive, chen2009direct, zhang2012new, zhai2018output, zhang2019low} and $p$-normal forms \cite{lin2000adding, qian2001non, qian2001continuous, lin2002adaptive, qian2002practical, qian2006recursive}. Motivated by these works, we address two problems in this paper. The first one is how to make classifications of lower triangular forms in favor of the design of control laws for these systems. The second problem is whether and how a nonlinear system can be equivalently transformed into a given type of lower triangular form.\par

Before discussing the classification scheme, we first review the related research on lower triangular forms. A nonlinear system is called a lower triangular form \cite{vcelikovsky1996equivalence} if it takes the form 
\begin{equation}\label{eq_tri}
\begin{aligned}
	&\dot{x}_{1}=f_{1}\left(x_{1}, x_{2}\right) \\
	&\vdots \\
	&\dot{x}_{n-1}=f_{n-1}\left(x_{1} ,\dots, x_{n}\right) \\
	&\dot{x}_{n}=f_{n}\left(x_{1}, \dots, x_{n}\right)+g_{n}\left(x_{1}, \dots, x_{n}\right) v
\end{aligned}
\end{equation}
where $x=(x_1,\dots,x_n)$ is the state vector, $v$ is the scale input, $g_n$ is a smooth function with $g_n(0) \ne 0$, and $f_i$, $i=1,\dots,n$, are smooth functions such that $\partial f_j / \partial x_{j+1} \not\equiv 0$, $j =1,\dots,n-1$, hold in a neighborhood of the origin. A lower triangular form is said to be a $p$-normal form \cite{cheng2003p, respondek2003transforming} if it is also of the special form
\begin{equation}\label{eq_p}
\begin{aligned}
	&{{\dot x}_1} = {\psi _{1,{p_1}}}({x_1},{x_2})x_2^{{p_1}} + \sum\limits_{j = 0}^{{p_1} - 1} {{\psi _{1,j}}({x_1})x_2^j} \\
	&\; \vdots \\
	&{{\dot x}_{n - 1}} = {\psi _{n - 1,{p_{n - 1}}}}({x_1}, \dots ,{x_n})x_n^{{p_{n - 1}}}\\
	&\ \ \ \ \ \ \ \ \ + \sum\limits_{j = 0}^{{p_{n - 1}} - 1} {{\psi _{n - 1,j}}({x_1}, \dots ,{x_{n - 1}})x_n^j} \\
	&{{\dot x}_n} = {f_n}({x_1}, \dots ,{x_n}) + {g_n}({x_1}, \dots ,{x_n})v
\end{aligned}
\end{equation}
where $p_i$, $i=1,\dots,{n-1}$, are positive integers, and $\psi_{i,j}$, $i=1,\dots,{n-1}$ and $j=p_i,\dots,1$, are smooth functions with 
\begin{equation}\nonumber
	\psi_{i, j}(0)\left\{\begin{matrix}
		\neq 0 & j=p_{i} \\
		=0 & j \neq p_{i}.
	\end{matrix}\right.
\end{equation}
When $p_1=\dots=p_{n-1}=1$, \eqref{eq_p} becomes a strict feedback form, which has been verified to be feedback equivalent to the controllable canonical form. The first report on $p$-normal forms was carried out by Lin and Qian. From 2000 to 2006, they conducted a series of systematic studies about the controller design for $p$-normal forms to meet various control objectives, including global stabilization  \cite{lin2000adding,qian2001non,qian2001continuous}, adaptive control \cite{lin2002adaptive}, output tracking \cite{qian2002practical}, and output feedback stabilization \cite{qian2006recursive}.  Subsequently, further impressive studies focused on those systems are presented, such as finite-time control \cite{hong2006adaptive,sun2015new,chen2020unified},  $ H_{\infty}$ control \cite{long2012h}, state-constrained control \cite{su2014stabilization},  global stabilization using multiple Lyapunov functions \cite{long2015integral}, nonsingular prescribed-time stabilization \cite{ding2020nonsingular}, and tracking control \cite{ding2021low}. 

Seeing that one can find a great many lower triangular forms other than $p$-normal forms and strict feedback forms, how to classify lower triangular forms is a problem worthy of study. As far as we know, there has been no report on this problem. Two classification schemes proposed in Section \uppercase\expandafter{\romannumeral3} are expected to be helpful in analyzing the behavior of lower triangular forms. The first classification scheme is directly inspired by $p$-normal forms. Let us denote the left-hand side of the $i$th equation of \eqref{eq_p} by $\varphi_i(x)$ for $i=1,\dots,n-1$. $P$-normal forms have a property that $\partial^{j} \varphi_{i} / \partial x_{i+1}^{j}(0)=0$, $j=1, \cdots, p_i-1$, and $\partial^{p_i} \varphi_i / \partial x_{i+1}^{p_i}(0) \neq 0$ are satisfied. In this paper, we say that $(0,\dots,0,\underbrace{p_{i}}_{(i+1) {\rm th}})$ is the least $(i+1)$-multi-index of $\varphi_i(x)$ (see Definition \ref{def_least}). These multi-indices are observed playing an key role in the controllers for $p$-normal forms  \cite{lin2000adding,qian2001non,qian2001continuous,lin2002adaptive,qian2002practical,qian2006recursive,hong2006adaptive,sun2015new,chen2020unified,long2012h,su2014stabilization,long2015integral,ding2020nonsingular,ding2021low}. This motivates us to classify \eqref{eq_tri} by the least $(i+1)$-multi-indices of $f_i(x)$ for $i=1,\dots,n-1$. Moreover, we will see that the least $(i+1)$-multi-index of $f_i(x)$ is invariant under a class of coordinate transformations called lower triangular coordinate transformations. The other way presented in this paper to classify lower triangular forms is based on another new notion called the greatest essential $(i+1)$-multi-index set of $f_i$ (see Definition \ref{def_we_e} and \ref{def_ge}). Since the least $(i+1)$-multi-index of $f_i$ belongs to the set, this classification is a refinement of the first one. It will be verified that the set is finite and invariant under any lower triangular coordinate transformation. Also, two algorithms for determining those sets from \eqref{eq_tri} are given in section \uppercase\expandafter{\romannumeral3}. It is reasonable to infer that those multi-indices can be expected to act as a pivotal part of the controllers for lower triangular forms, considering that the terms corresponding to the least $(i+1)$-multi-index of $f_i$ and the elements of the greatest essential $(i+1)$-multi-index set of $f_i$ can be regarded as the invariant "control" terms for the equation $\dot{x}=f_i(x_1,\dots,x_{i+1})$ given in \eqref{eq_tri} (see Remark \ref{rem_controlterms}). \par
Since a classification of lower triangular forms induces a classification of all the systems that are feedback equivalent to lower triangular forms, the next problem naturally considered in this paper is whether a given nonlinear system is equivalent to a specific type of lower triangular form via a state feedback and a change of coordinates. This problem is about feedback equivalence between different systems. In recent decades, a series of original results have been achieved on the issue of feedback equivalence. In 1973, Krener provided several sufficient and necessary conditions that an affine nonlinear system is equivalent to another affine system or a linear system via a local coordinate transformation \cite{krener1973equivalence}. In 1978, taking invariants under feedback into consideration, Brockett proposed a necessary and sufficient condition for a nonlinear system to be equivalent to a controllable linear system via a local coordinate transformation $x=T(\xi)$ and a state feedback of the form  $u = \alpha_u(\xi) + \beta_u v$, where $x$ and $\xi$ are two state vectors, $ \alpha_u(\xi)$ is a smooth function, and $\beta_u$ is a real number \cite{brockett1978feedback}. In the 1980s, the problem of exact linearization with a feedback taking the form $u = \alpha_u(\xi) + \beta_u(\xi) v$, where $\beta_u(\xi)$ is a function satisfying $\beta_u(\xi) \ne 0$, was solved in \cite{jacubczyk1980linearization,su1982On,hunt1983global}. The multi-input exact feedback linearization problem was solved in \cite{isidori1981nonlinear}. Cheng and Lin  \cite{cheng2003p} presented a necessary and sufficient condition under which a nonlinear system is feedback equivalent to a $p$-normal form via a coordinate transformation and a state feedback of the form $u = \alpha_u(\xi) + \beta_u v$, and also designed an algorithm to find the appropriate coordinate transformations and feedback control laws in 2003. In late this year, Respondek \cite{respondek2003transforming} solved the $p$-normalization problem  using a state feedback of the form $u = \alpha_u(\xi) + \beta_u(\xi) v$ and pointed out $p$-normal forms, taking the form \eqref{eq_p}, are all locally equivalent to their special cases with $\psi_{i, p_i}(x) = 1$ for $i = 1,\dots,n-1$. \par
  Two methods are provided to determine whether a nonlinear system is feedback equivalent to a given type of lower triangular form in Section \uppercase\expandafter{\romannumeral4}. A way to solve the problem is to transform the system into a lower triangular system from which one then can determine the least $(i+1)$-multi-index and the greatest essential $(i+1)$-multi-index set of the right-hand side of its $i$th equation. A new necessary and sufficient condition for a single-input nonlinear system to be equivalent to a lower triangular form is given to simplify the transformation mentioned above. Since it may be quite difficult to find an appropriate change of coordinates to transform a system into a lower triangular form, we seek a new method for judging the type without implementing an equivalent transformation. Theorem \ref{theo_sysleast}, Theorem \ref{theorem_type}, Corollary \ref{col_least}, and Corollary \ref{col_type} allow us to determine whether a nonlinear system is equivalent to a specific type of lower triangular form by computing Lie brackets.\par

The rest of this paper is organized as follows. Section \uppercase\expandafter{\romannumeral2} will describe in detail the problem of how to classify single-input lower triangular forms and the problem of whether a system is equivalent to a specific type of lower triangular form. Section \uppercase\expandafter{\romannumeral3} gives two ways to solve the former, and Section \uppercase\expandafter{\romannumeral4} discusses the latter. We conclude the paper in Section  \uppercase\expandafter{\romannumeral5}.

\section{Problem Formulations}
To begin with, we clarify that throughout this paper all the definitions and statements are local, although it is possible to generalize to the global as well. In other words, we always operate in some neighborhoods of the origin which are small enough. To classify lower triangular forms, we pay special attention to a class of coordinate transformations defined as follows. 
\begin{definition}\label{def_tri_coord}
	A local coordinate transformation $y = U(x)$ is said to be lower triangular if it takes the form 
	\begin{equation}\label{eq_tri_coord}
		\begin{aligned}
			&{y_1} = {U_1}\left( {{x_1}} \right)\\
			&\; \vdots \\
			&{y_n} = {U_n}\left( {{x_1}, \dots ,{x_n}} \right).
		\end{aligned}
	\end{equation}
\end{definition}\par
\begin{lemma}\label{lem_tri_coord}
Let $y = U(x)$ be a coordinate transformation. Rewriting \eqref{eq_tri} in $y$-coordinates, it still takes a lower triangular form if and only if the coordinate transformation is of the form \eqref{eq_tri_coord}. Moreover, the inverse transformation of \eqref{eq_tri_coord} is also a local lower triangular coordinate transformation.
\end{lemma}

The classifications we investigate here should guarantee that the type a lower triangular form belongs to is unchanged under any lower triangular coordinate transformation.

There are some clarifications about the classifications of lower triangular forms we would like to illustrate. First, the rules we design to classify lower triangular forms are independent of $f_n(x)$ and $g_n(x)$ introduced in \eqref{eq_tri} because they can be changed by the input $v$. Suppose that ${f'_n}(x)$ and ${g'_n}(x)$ are two given smooth functions with ${g'_n}(0) \ne 0$. Take
$v = {{\left( {{{f'}_n}(x) - {f_n}(x)} \right)} \mathord{\left/
		{\vphantom {{\left( {{{f'}_n}(x) - {f_n}(x)} \right)} {{g_n}(x)}}} \right.
		\kern-\nulldelimiterspace} {{g_n}(x)}} + {{{{g'}_n}(x)} \mathord{\left/
		{\vphantom {{{{g'}_n}(x)} {{g_n}(x)}}} \right.
		\kern-\nulldelimiterspace} {{g_n}(x)}}v'$  
in an appropriate neighborhood of the origin, and then the last equation of \eqref{eq_tri} becomes  
${\dot x_n} = {f'_n}(x) + g'(x)v'$. Second, in some literature, such as \cite{sun2015new, chen2020unified,ding2020nonsingular}, the parameters $p_i$, $i=1,...,n-1$, in \eqref{eq_p} are allowed to be selected as positive fractions. Since $x_{i+1}^{p_i}$ is not smooth at the origin when $p_i$ is not a nonnegative integer, we only consider the case that  $p_i$, $i=1,...,n-1$, are all positive integers. Last, a smooth nonaffine system 
\begin{equation}\nonumber
\begin{aligned}
	&\dot{x}_{1}=f_{1}\left(x_{1}, x_{2}\right) \\
	&\vdots \\
	&\dot{x}_{n-1}=f_{n-1}\left(x_{1}, \dots, x_{n}\right) \\
	&\dot{x}_{n}=f_{n}\left(x_{1}, \dots, x_{n}, v\right),
\end{aligned}
\end{equation}
can be equivalently transformed into an affine system via adding a new coordinate variable $x_{n+1}=v$. In fact, the system can be rewritten as
\begin{equation}\nonumber
\begin{aligned}
	&\dot{x}_{1}= f_{1}\left(x_{1}, x_{2}\right) \\
	&\vdots \\
	&\dot{x}_{n}= f_{n}\left(x_{1}, \dots, x_{n+1}\right) \\
	&\dot{x}_{n+1} = \dot{v}.
\end{aligned}
\end{equation} 
Thus, a classification of affine lower triangular forms can be naturally extended to nonaffine lower triangular forms, and we only examine affine systems here.

If the problem of how to classify lower triangular forms has been solved, let us consider a single-input nonlinear system 
\begin{equation}\label{eq_s}
\begin{aligned}
	&{{\dot \xi }_1} = {F_1}(\xi ) + {G_1}(\xi )u\\
	&\dots \\
	&{{\dot \xi }_n} = {F_n}(\xi ) + {G_n}(\xi )u
\end{aligned}
\end{equation} 
where $\xi  = {({\xi _1}, \dots ,{\xi _n})} \in {{\mathbb{R}}^n}$ is the system state, $u \in {\mathbb{R}}$ is the control input, $F_i(\xi)$, $i=1,\dots,n$, are smooth functions with $F_i(0)=0$, and $G_i(\xi)$, $i=1,\dots,n$, are all smooth functions such that there exists an integer $j \in \{1,\dots,n\}$ satisfying $G_j(0) \ne 0$. The next problem we address in this paper is whether \eqref{eq_s} is locally equivalent to a given type of lower triangular form via a state feedback and a change of coordinates. The state feedback considered here is of the form
\begin{equation}\label{eq_fb}
	u = \alpha_u (\xi ) + \beta_u (\xi )v
\end{equation} 
where $\alpha_u (\xi )$ and $\beta_u (\xi )$ are smooth functions with $\beta_u (0) \neq  0$, and the change of coordinates can be expressed as 
\begin{equation}\label{eq_ct}
	x = T(\xi ) = {\left( {{T_1}(\xi ), \dots ,{T_n}(\xi )} \right)}
\end{equation} 
where $T:{{\mathbb{R}}^n} \to {{\mathbb{R}}^n}$ is a smooth invertible mapping with $T(0)=0$. 

\section{Classifications of Lower Triangular Forms}
The problem we are concerned with in this section is how to classify lower triangular forms. Let us start with the following two definition.
\begin{definition}
An $m$-dimensional multi-index or $m$-multi-index is an ordered $m$-tuple 
\begin{equation}\label{eq_multi-index}
\alpha  = ({\alpha _1} \dots ,{\alpha _m})
\end{equation} 
where $m$ is an integer satisfying $1 \le m \le n$ and ${\alpha _i}$, $i = 1, \dots ,m$, are all nonnegative integers \cite{rudin1991functional}; \eqref{eq_multi-index} is called a proper $k$-multi-index if ${\alpha _k} \ge 1$ and ${\alpha _{k + 1}} =  \dots  = {\alpha _m} = 0$ hold for some $k$ with $1 \le k \le m$; \eqref{eq_multi-index} is said to be a proper $0$-multi-index if $\alpha_i  = 0$ for all $i=1,\dots,m$, and we may simply write $\alpha = 0$ in this case.
\end{definition}\par
\begin{definition}
	Let $\alpha$ and $\beta$ be multi-indices. We write $\alpha=\beta$ if and only if they are both proper $k$-multi-indices with $k \ge 0$ and $\alpha_i=\beta_i$ holds for every $i = 1,\dots ,k$ when $k>0$ \cite{rudin1991functional}.
\end{definition}\par
\begin{remark}
	Every proper $k$-multi-index can be regarded as an $i$-multi-index with $i \ge k$.
\end{remark}\par
 Taking $\alpha$ as an $m$-dimensional multi-index, for ease of notation, we write 
\begin{equation}\nonumber
	{x^\alpha } = x_1^{\alpha_1} \dots x_m^{\alpha_m}
\end{equation} 
and 
\begin{equation}\nonumber
	\frac{{{\partial ^\alpha }}}{{\partial {x^\alpha }}} = \frac{{{\partial ^{\left| \alpha  \right|}}}}{{\partial x_1^{\alpha _1} \dots \partial x_m^{\alpha _m}}}
\end{equation} 
where $\left| \alpha  \right|=\alpha_1+ \dots +\alpha_m$. Moreover, if $p(x_1,\dots,x_m)$ is a function and $\left| \alpha  \right|=0$, we define that $\partial ^{\left| \alpha  \right|} p/{\partial x_1^{\alpha _1} \dots \partial x_m^{\alpha _m}}=p(x_1,\dots,x_m)$ \cite{rudin1991functional}.

\begin{definition}
	$p(x_1,\dots,x_m)$ is a smooth function (or a holomorphic function) and $\alpha$ is a multi-index with $\left| \alpha \right| > 0$. We say that $\alpha$ is a multi-index of $p(x_1,\dots,x_m)$ (with respect to the coordinates $x_1,\dots,x_m$) if $\partial ^\alpha p / \partial x^\alpha (0) \neq 0$ holds.
\end{definition}\par
\begin{remark}
	$0$ is a multi-index of $p$ if and only if $p(0) \ne 0$.
\end{remark}
\begin{remark}
	In most cases, we consider the function $p(x_1,\dots,x_m)$ to be real-valued and smooth. This function is allowed to be complex-valued and holomorphic only for discussing invariant multi-indices in subsection B. For the same reason, the lower triangular coordinate transformation $y=U(x)$ defined by Definition \ref{def_tri_coord} can be smooth or biholomorphic.
\end{remark}
\begin{proposition}
	Suppose $p(x_1,\dots,x_m)$ is a smooth (or a holomorphic) function and $\alpha$ is a multi-index of $p(x_1,\dots,x_m)$. $p$ can be express as 
	\begin{equation}\nonumber
		p(x_1,\dots,x_m) = {c_\alpha }{x^\alpha } + \bar p(x_1,\dots,x_m).
	\end{equation}
	In above equation, $c_\alpha = {{{\partial ^\alpha p}}}/{{\partial {x^\alpha }}}(0)$ is a nonzero coefficient and $\bar p(x_1,\dots,x_m)$ is a function with ${{{\partial ^\alpha \bar p}}}/{{\partial {x^\alpha }}}(0) = 0$.
\end{proposition}\par
For convenience, let us denote the set of all the proper $k$-multi-indices of $p(x_1,\dots,x_m)$ by ${\cal I}_k (p)$ for $k=0,\dots,m$ and write ${\cal I} (p) = \bigcup_{k=0}^{m} {\cal I}_k (p)$ throughout this paper.\par
The rest of this section is divided into three subsections. Subsection A discusses several properties of multi-indices. Subsection B investigates the invariant multi-indices of a function under lower triangular coordinate transformations. In Subsection C, we propose two classification schemes of lower triangular forms.\par

\subsection{The Least Multi-index and Essential Multi-indices of Functions}
In this subsection, we investigate which multi-indices of a function may be more vital by exploring the relations between multi-indices. The following definition presents one of the ways to compare two multi-indices.
\begin{definition}
	$\alpha$ and $\beta$ are proper $k_\alpha$-multi-index and proper $k_\beta$-multi-index, respectively. Let $m = \mathrm{max}(k_\alpha,k_\beta) + 1$. We say that $\alpha$ is less than $\beta$ in lexicographical order, denoted by $\alpha \lessdot \beta $, if there exists an integer $i \in \{1,\dots,m\}$ such that $\alpha_i < \beta_i$ and $\alpha_j = \beta_j$ for all $j = 1,\dots,i-1$.
\end{definition}
\begin{example}
	As defined above, we have  $(2,3,9) \lessdot (2,5,1)$ and $(0,3) \lessdot
	(1,0,1)$.
\end{example}

\begin{definition}\label{def_least}
	 Let $I$ be a set whose members are all proper $i$-multi-indices. $\alpha \in I$ is said to be the least $i$-multi-index of $I$ if $\alpha \lessdot \beta$ holds for any $\beta \in I$ different from $\alpha$. Further let $p(x_1,\dots,x_m)$ be a smooth function (or a holomorphic function). We also call the least $i$-multi-index of $ {\cal I}_i (p)$ as the least $i$-multi-index of $p$, written as ${\cal L}_i(p)$.
\end{definition}
\begin{remark}
	${\cal L}_0(p)=0$ if and only if $p(0) \ne 0$. 
\end{remark}
\begin{example}
	Consider the following lower triangular form
	\begin{equation}\nonumber
		\begin{aligned}
			&{{\dot x}_1} = x_1 x_2^3+x_1^2 x_2 - {x_1} = f_1(x_1,x_2)\\
			&{{\dot x}_2} = x_2^2{x_3} + x_1 x_3  = f_2(x_1,x_2,x_3)\\
			&{{\dot x}_{\rm{3}}} = {x_3} + u\; .
		\end{aligned}
	\end{equation}
	We have ${\cal L}_2(f_1)=(1,3)$ and ${\cal L}_3(f_2)=(0,2,1)$.
\end{example}
\begin{definition}\label{def_prec}
	Let $\alpha$ and $\beta$ be multi-indices. If there exists a lower triangular  coordinate transformation $y = U(x)$ such that 
	\begin{equation}\nonumber
		{x^\alpha } = {c_\beta }{y^\beta } + h(y)
	\end{equation}
where $c_\beta \neq 0$ and the function $h(y)$ satisfies ${{{\partial ^\beta }h} / {\partial {y^\beta }}}(0) = 0$, then we say that $\beta$ is generated by $\alpha$, denoted by $\alpha \preceq \beta $. If $\alpha \preceq \beta $ and $\alpha \ne \beta $, we write $\alpha \prec \beta $.
\end{definition}
\begin{remark}
	Arbitrary proper $i$-multi-index ($i>0$) can be generated by the proper $i$-multi-index $(0,\dots,0,1)$. The $0$-multi-index $0$ can only generate itself and can only be generated by itself.
\end{remark}

\begin{example}
	Let $\alpha=(1,2,1)$, and select the following lower triangular coordinate transformation
	\begin{equation}\nonumber
		\begin{aligned}
			&{y_1} = {x_1}\\
			&{y_2} = {x_2} + x_1^2\\
			&{y_3} = {x_3} + {x_1},
		\end{aligned}
	\end{equation}
whose inverse transformation can be expressed as 
	\begin{equation}\nonumber
		\begin{aligned}
			&{x_1} = {y_1}\\
			&{x_2} = {y_2} - y_1^2\\
			&{x_3} = {y_3} - {y_1}.
		\end{aligned}
	\end{equation}
Substituting the above equations into $x^\alpha$ yields  
	\begin{equation}\nonumber
		\begin{aligned}
		{x^\alpha } &= {x_1}x_2^2{x_3} \\
		&= {y_1}y_2^2{y_3} - y_1^2y_2^2 - 2y_1^3{y_2}{y_3} + 2y_1^4{y_2} + y_1^5{y_3} - y_1^6 \; ;
		\end{aligned}
	\end{equation}
that is, $\alpha$ can generate at least the six 3-multi-indices as follows: $(1,2,1)$, $(2,2,0)$, $(3,1,1)$, $(4,1,0)$, $(5,0,1)$, and $(6,0,0)$. 
\end{example}

\begin{proposition}\label{prop_ngene}
	Let $\alpha$ and $\beta$ be proper $m_\alpha$-multi-index and proper $m_\beta$-multi-index, respectively. If $m_\alpha<m_\beta$ then $\alpha$ can not generate $\beta$. 
\end{proposition}

\begin{theorem}\label{thm_gene_condi}
	$\alpha$ and $\beta$ are proper $m_\alpha$-multi-index and proper $m_\beta$-multi-index, respectively,  satisfying $m_\alpha \geq m_\beta > 0$ and $\alpha \neq \beta$. Then $\alpha \prec \beta $  if and only if for all $i=1,\dots,m_\alpha$ we have
	\begin{equation}\label{neq_prec}
		\sum\limits_{j = 1}^i {{\alpha _j}}  \le \sum\limits_{j = 1}^i {{\beta _j}} .
	\end{equation}
\end{theorem}

\begin{proof}
	The necessity is obvious, let us verify the sufficiency. We first consider the case of  $m_\alpha = 1$. It is clear that $m_\beta = 1$ in this case. From $\alpha \ne \beta$ and \eqref{neq_prec}, $0<\alpha_1 < \beta_1$ holds. Let $h(x_1) = x_1^{\alpha_1}$ and $y_1$ a new coordinate satisfying
	\begin{equation}\nonumber
		{x_1} = {y_1} + y_1^{{\beta _1} - {\alpha _1} + 1}\;.
    \end{equation}	
	Since substituting the above equation into $h(x_1)$ yields
	\begin{equation}\nonumber
		h(x_1)=h_y(y_1) = \sum\limits_{i = 0}^{{\alpha _1}} {\left( {\begin{matrix}
					{{\alpha _{_1}}}\\
					i
			\end{matrix}} \right)y_1^{{\alpha _1} - i}y_1^{i({\beta _1} - {\alpha _1} + 1)}} \; ,
	\end{equation} 
	we have
	\begin{equation}\nonumber
		\frac{{{\partial ^{{\beta _1}}}h_y}}{{\partial {y_1^{{\beta _1}}}}}(0) = {\alpha_1} \cdot {\beta_1}! \ne 0\;,
	\end{equation}	
	and then $\alpha \prec \beta$ holds for the case.	 
	
	Suppose that, for an integer $k>0$ and all the $m_\alpha =1,\dots,k$, $\alpha \prec \beta$ holds when  \eqref{neq_prec} is satisfied. We now prove that, \eqref{neq_prec} still implies $\alpha \prec \beta$ when $\alpha$ is a proper $(k+1)$-multi-index. To this end, let us consider the two cases as discussed below. For the case ${\alpha_{k+1}} < {\beta_{k+1}}$, one can construct a family of new coordinates $y_1,\dots,y_k$ satisfying
	\begin{equation}\nonumber
		\begin{aligned}
			&{x_1} = {U_1}\left( {{y_1}} \right)\\
			&\;\vdots \\
			&{x_{k}} = {U_{k}}\left( {{y}, \dots ,{y_{k}}} \right).
		\end{aligned}
	\end{equation}
	and $x_1^{\alpha_1} \dots x_k^{\alpha_k} = c_{(\beta_1,\dots,\beta_k)} y_1^{\beta_1} \dots y_k^{\beta_k} + s(y)$ where the coefficient $c_{(\beta_1,\dots,\beta_k)} \ne 0$ and  $\partial^{(\beta_1,\dots,\beta_k)} s / \partial y^{(\beta_1,\dots,\beta_k)} (0) = 0$. If we choose the next coordinate $y_{k+1}$ satisfying
	\begin{equation}\nonumber
		{x_{k+1}} = {y_{k+1}} + y_{k+1}^{{\beta _{k+1}} - {\alpha _{k+1}} + 1}
	\end{equation}
	then 
	\begin{equation}\label{eq_case1}
		\begin{aligned}
			x^\alpha = & \left( c_{(\beta_1,\dots,\beta_k)} y_1^{\beta_1}  \dots y_k^{\beta_k} + s(y)\right)  \; \cdot \\
			&   \sum\limits_{i = 0}^{{\alpha _{k+1}}} {\left( {\begin{matrix}
						{{\alpha _{_{k+1}}}}\\
						i
				\end{matrix}} \right)y_{k+1}^{{\alpha _{k+1}} - i}y_{k+1}^{i({\beta _{k+1}} - {\alpha _{k+1}} + 1)}}
		\end{aligned}
	\end{equation}
	is obtained. There is a term  
	\begin{equation}\nonumber
		c_{(\beta_1,\dots,\beta_k)}\, \cdot \, \alpha_{k+1} \, \cdot \, y_1^{\beta_1} \dots y_{k+1}^{\beta_{k+1}}
	\end{equation}	
	in the right-hand side of \eqref{eq_case1}, which implies $\alpha \prec \beta$. The other case is ${\alpha_{k+1}} \ge {\beta_{k+1}}$. Taking a family of new coordinates $\bar x_1,\dots,\bar x_{k+1}$ satisfying
	\begin{equation}\nonumber
	\begin{aligned}		
		&{x_1} = V_1(\bar x_1) = {{\bar x}_1}\\
		&\; \vdots\\
		&{x_k} = V_k(\bar x_k) = {{\bar x}_k}\\
		&{x_{k + 1}} = V_{k+1}(\bar x) ={{\bar x}_k} + {{\bar x}_{k + 1}},
	\end{aligned}
	\end{equation}	
	denoted it by $( x_1,\dots,x_{k+1})=V(\bar x_1,\dots,\bar x_{k+1})$, we compute the function $h_1(\bar x) = x^\alpha$. 	
	\begin{equation}\nonumber
		h_1(\bar x) = \bar x_1^{{\alpha _1}} \dots \bar x_{k - 1}^{{\alpha _{k - 1}}}\sum\limits_{i = 0}^{{\alpha _{k + 1}}} {\left( {\begin{matrix}
					{{\alpha _{k + 1}}}\\
					i
			\end{matrix}} \right)\bar x_k^{{\alpha _k} + i}\bar x_{k + 1}^{{\alpha _{k + 1}} - i}} 
	\end{equation}
	The right-hand side of the above equation includes the term	
	\begin{equation}
		\left( {\begin{matrix}
				{{\alpha _{k + 1}}}\\
				{{\alpha _{k + 1}} - {\beta _{k + 1}}}
		\end{matrix}} \right)\bar x_1^{{\alpha _1}} \dots \bar x_{k - 1}^{{\alpha _{k - 1}}}\bar x_k^{({\alpha _k} + {\alpha _{k + 1}} - {\beta _{k + 1}})}\bar x_{k + 1}^{{\beta _{k + 1}}}.
	\end{equation}
	Let us denote the multi-index of this term by $\gamma  =( \gamma_1,\dots,\gamma_{k+1}) = ({\alpha _1},\dots,{\alpha _{k - 1}},{\alpha _k} + {\alpha _{k + 1}} - {\beta _{k + 1}},{\beta _{k + 1}}) $, it is obvious that $\alpha \preceq  \gamma$. Additionally, for any $l = 1, \dots ,k$, the inequality $\textstyle{\sum_{j = 1}^l {{\gamma _j}}} \le \textstyle{ \sum_{j = 1}^l {{\beta _j}}}$ holds, and then  $(\gamma_1,\dots,\gamma_k) \prec (\beta_1,\dots,\beta_k)$ is obtained. This means that we can find a new family of coordinates $(y_1, \dots, y_{k+1})$ satisfying 
	\begin{equation}\nonumber
		\begin{aligned}
			&{{\bar x}_1} = {W_1}\left( {{y_1}} \right)\\
			&\; \vdots \\
			&{{\bar x}_k} = {W_k}\left( {{y_1}, \dots ,{y_k}} \right)\\
			&{{\bar x}_{k + 1}} = {W_{k+1}(y_{k + 1})}= {y_{k + 1}}
		\end{aligned}
	\end{equation}	
	 such that $\beta$ is a multi-index of $h_2(y) ={\bar x}^\gamma = (W_1(y),\dots,W_{k+1}(y))^\gamma=W(y)^\gamma $ with respect to $y$-coordinates; that is, $\gamma \prec \beta$. To prove $\alpha \prec \beta$, let $\delta \neq \gamma$ be another multi-index of $h_1(\bar x)$ with respect to $\bar{x}$. Since  $\delta_{k+1} \neq \beta_{k+1}$,  $\beta$ is not a multi-index of $h_3(y) = \bar x ^ \delta = W(y)^\delta$ with respect to $y$-coordinates. Thus, $\beta$ must be a multi-index of $h_4(y)=x^\alpha=(V(W(y)))^\alpha$ with respect to $y$-coordinates. 
	 
	 Therefore, $\alpha \prec \beta$ holds when we have \eqref{neq_prec}.
\end{proof}

From the above theorem, the following two corollaries are immediate consequences. 
\begin{corollary}
	Let $I$ be a set of multi-indices,  and take any $\alpha, \beta, \gamma \in I$. Then the relation $\preceq$ has the following properties: \par
	(i) $\alpha \preceq \alpha$;\par
	(ii) both $\alpha \preceq \beta$ and $\beta \preceq \alpha$ imply $\alpha = \beta$; \par 
	(iii) both $\alpha \preceq \beta$ and $\beta \preceq \gamma$ imply $\alpha \preceq \gamma$. \\
	That is, $\preceq$ is a partial order on the ground set $I$.  
\end{corollary}

\begin{corollary}
	$\alpha$ and $\beta$ are proper $m_\alpha$-multi-index and proper $m_\beta$-multi-index, respectively, with $m_\alpha \geq m_\beta >0$. $\alpha \nprec \beta$ if and only if for some $i \in \{1,\dots,m_\alpha \}$ the inequality
	\begin{equation}\nonumber
		\sum\limits_{j = 1}^i {{\alpha _j}}  > \sum\limits_{j = 1}^i {{\beta _j}}. 
	\end{equation}
	holds. In the case of $m_\alpha = m_\beta$, $\alpha \npreceq \beta$ and $\beta \npreceq \alpha$ are both true if and only if there exist two integers  $i_1 , i_2 \in \{1,\dots,m_\alpha \}$ such that the following inequalities hold.
	\begin{equation}\nonumber
		\sum\limits_{j = 1}^{{i_{\rm{1}}}} {{\alpha _j}}  < \sum\limits_{j = 1}^{{i_{\rm{1}}}} {{\beta _j}} ,{\kern 1pt} \sum\limits_{j = 1}^{{i_{\rm{2}}}} {{\alpha _j}}  > \sum\limits_{j = 1}^{{i_{\rm{2}}}} {{\beta _j}}
	\end{equation}
\end{corollary}

\begin{definition}\label{def_we_e}
	Let $I$ be a set of multi-indices and $\alpha \in I$ a proper $i$-multi-index. $\alpha$ is said to be a weakly essential $i$-multi-index of $I$ if there is no another proper $i$-multi-index of $I$ that can generate $\alpha$. If $\alpha' \nprec \alpha$ holds for any $\alpha' \in I$, we say that $\alpha$ is an essential $i$-multi-index of $I$. $p(x_1,\dots,x_m)$ is a smooth function (or a holomorphic function) and $\beta$ is a proper $i$-multi-index of $p$. $\beta$ is said to be a weakly essential $i$-multi-index of $p$ if $\beta$ is a weakly essential $i$-multi-index of ${\cal I} (p)$. Moreover, if $\beta$ is an essential $i$-multi-index of ${\cal I}(p)$, we say that $\beta$ is an essential $i$-multi-index of $p$. 
\end{definition}

\begin{lemma}
	$p(x_1,\dots,x_m)$ is a smooth function (or a holomorphic function) and $x = V(y)$ is a lower triangular coordinate transformation. For an $m$-multi-index $\alpha  = ({\alpha _1}, \dots ,{\alpha _m}) \ne 0$ and the function $q(y)=p(V(y))$, we have
	\begin{equation}\label{eq_pd}
		\begin{aligned}
			\frac{{{\partial ^\alpha }q}}{{\partial {y^\alpha }}} = \sum\limits_{\beta  \preceq \alpha } {\sum\limits_{\alpha  = \sum\limits_{k,i} {{\gamma ^{\beta,k,i}}} } {\left( {\frac{{{\partial ^\beta }q}}{{\partial {x^\beta }}}\prod\limits_{k = 1}^m {\prod\limits_{i = 1}^{{\beta _k}} {\frac{{{\partial ^{{\gamma ^{\beta,k,i}}}}{x_k}}}{{\partial {y^{{\gamma ^{\beta,k,i}}}}}}} } } \right)} } 
		\end{aligned}
	\end{equation}
	where $\beta = (\beta_1,\dots,\beta_m)$ and every $\gamma ^{\beta,k,i}$ is a $k$-multi-index.
\end{lemma}
\begin{proof}
	Let $\alpha =(0,\dots,0,\alpha_i,0,\dots,0)= (0,\dots,0,1,0,\dots,0)$ where $1 \le i \le m$, then
	\begin{equation}\nonumber
		\frac{{{\partial ^{\alpha}}q}}{{\partial {y^\alpha }}} = \frac{{\partial q}}{{\partial {y_i}}} = \frac{{\partial q}}{{\partial {x_i}}}\frac{{\partial {x_i}}}{{\partial {y_i}}} +  \dots  + \frac{{\partial q}}{{\partial {x_m}}}\frac{{\partial {x_m}}}{{\partial {y_i}}}
	\end{equation}
The equation above implies that \eqref{eq_pd} is satisfied in this case. 

Assume \eqref{eq_pd} holds for a nonzero multi-index $\alpha  = ({\alpha _1}, \dots ,{\alpha _m}) = (0, \dots ,0,{\alpha _j}, \dots ,{\alpha _m})$, where $1 \leq j \leq m$ and $\alpha_j \ge 0$. Let $\alpha ' = (0, \dots ,0,{\alpha _j}{\rm{ + 1}},{\alpha _{j + 1}}, \dots ,{\alpha _m})$. For all $\beta$ satisfying $\beta \preceq \alpha$, $\beta \preceq \alpha'$ can be deduced by using Theorem \ref{thm_gene_condi}. We now focus on the case of $\bar\beta \preceq \alpha'$ but $\bar\beta \npreceq \alpha$. There exists an integer $k \in \{j,\dots,m\}$ such that $\textstyle{\sum_{i = 1}^k {{\bar\beta _i}}}  > \textstyle{\sum_{i = 1}^k {{\alpha _i}}} $ and $\textstyle{\sum_{i = 1}^l {{\bar\beta_i}}}  \le \textstyle{\sum_{i = 1}^l {{\alpha _i}}} $ for all $l=1,\dots,k-1$. Comparing $\alpha'$ to $\alpha$, the relation
\begin{equation}\nonumber
	(0, \dots ,0,{\bar\beta_j}, \dots ,{\bar\beta_{k - 1}},{\bar\beta_k} - 1,{\bar\beta_{k + 1}}, \dots ,{\bar\beta_m}) \preceq \alpha
\end{equation}
must hold for this case. Then a direct calculation presented by \eqref{eq_pdc} shows that \eqref{eq_pd} holds for $\alpha'$.
	\begin{figure*}[!t]
	\begin{equation}\label{eq_pdc}
		\begin{aligned}
			&\frac{\partial^{\alpha^{\prime}} q}{\partial y^{\alpha^{\prime}}}=\frac{\partial^{\alpha^{\prime}} q}{\partial y_{m}^{\alpha_{m}} \cdots \partial y_{j+1}^{\alpha_{j+1}} \partial y_{j}^{\alpha_{j}+1}}=\partial\left(\sum_{\beta \preceq \alpha} \sum_{\alpha=\sum_{k, i}^{\beta, k, i}}\left(\frac{\partial^{\beta} q}{\partial x_{m}^{\beta_{m}} \cdots \partial x_{j}^{\beta_{j}}} \prod_{k=j}^{m} \prod_{i=1}^{\beta_{k}} \frac{\partial^{\gamma^{\beta, k, i}} x_{k}}{\partial y_{k}^{\gamma_{k}^{\beta, k, i}} \cdots \partial y_{j}^{\gamma_{j}^{\beta, k, i}}}\right)\right) \large{\bigg /} \partial y_{j}\\
			&\quad =\sum_{\beta \preceq \alpha} \sum_{\alpha=\sum_{k, i}^{\beta, k, i}}\left(\left(\sum_{l=j}^{m} \frac{\partial^{\beta} q}{\partial x_{m}^{\beta_{m}} \cdots \partial x_{l}^{\beta_{l}+1} \cdots \partial x_{j}^{\beta_{j}}} \frac{\partial x_{l}^{\beta_{l}+1}}{\partial y_{j}}\right) \prod_{k=j}^{m} \prod_{i=1}^{\beta_{k}} \frac{\partial^{\gamma^{k, i}} x_{k}}{\partial y_{k}^{\gamma_{k}^{k, i}} \cdots \partial y_{j}^{\gamma_{j}^{k, i}}}+\frac{\partial^{\beta} q}{\partial x_{m}^{\beta_{m}} \cdots \partial x_{j}^{\beta_{j}}} \cdot\right. \\
			&\qquad \left.\partial\left(\prod_{k=j}^{m} \prod_{i=1}^{\beta_{k}} \frac{\partial^{\gamma^{\beta, k, i}} x_{k}}{\partial y_{k}^{\gamma_{k}^{k, i}} \cdots \partial y_{j}^{\gamma_{j}^{k, i}}}\right) \large{\bigg/} \partial y_{j}\right)=\sum_{\beta^{\prime} \preceq \alpha^{\prime}} \sum_{\alpha^{\prime}=\sum_{k, i} \lambda^{\beta^{\prime}, k, i}}\left(\frac{\partial^{\beta^{\prime}} q}{\partial x_{m}^{\beta_{m}^{\prime}} \cdots \partial x_{j}^{\beta_{j}^{\prime}}} \prod_{k=j}^{m} \prod_{i=1}^{\beta_{k}^{\prime}} \frac{\partial^{\gamma^{\beta^{\prime}, k, i}} x_{k}}{\partial y_{k}^{\lambda_{k}^{\beta^{\prime}, k, i}} \cdots \partial y_{j}^{\lambda_{j}^{\beta^{\prime}, k, i}}}\right)
		\end{aligned}
	\end{equation}
	\hrulefill
	\end{figure*}
	This proves \eqref{eq_pd}. 	
\end{proof}

\begin{proposition}\label{prop_wess_inv}
	$\alpha$ is a weakly essential $i$-multi-index of a smooth function (or a holomorphic function) $p(x_1,\dots,x_m)$ if and only if $\alpha$ is still a weakly essential $i$-multi-index of the function $q(y)=p(V(y))$ where $x=V(y)$ is a coordinate transformation taking the form 
	\begin{equation}\label{eq_tri_coord_i}
		\begin{aligned}
			&x_1 = V_1(y_1),\dots,x_i = V_i(y_1,\dots,y_i),\\
			&x_{i+1}=V_{i+1}(y_{i+1})=y_{i+1},\dots,x_m=V_m(y_m)=y_m.
		\end{aligned}
	\end{equation} 
\end{proposition}
\begin{proof}
	Necessity. Since when $i=0$ the necessity is obvious, we only consider the case of $i \ge 1$. Let $\beta$ be a multi-index satisfying $\beta \prec \alpha$. If $\beta$ is a proper $i$-multi-index, we obtain ${{{\partial ^\beta }q}/{\partial {x^\beta }(0)}} = 0$ since $\beta$ is not a multi-index of $p(x)$. Now consider $\beta$ as a proper $i'$-multi-index with $i < i' \le m$. Owing to $\partial x_{i'}/\partial y_k=0$ for any $k = 1,\dots,i$, any term in $\partial ^\alpha q/\partial y^\alpha$ which has the multiplier $\partial ^\beta q/\partial y^\beta$ is equal to 0. Thus, the only term in ${{\partial ^\alpha }q}/{\partial {y^\alpha }}$ that is not equal to 0 at the origin is
	\begin{equation}\nonumber
		\frac{{{\partial ^\alpha }q}}{{\partial {x^\alpha }}}
		\left( \frac{\partial x_1}{\partial y_1} \right)^{\alpha_1} \dots
		\left( \frac{\partial x_m}{\partial y_i} \right)^{\alpha_i};
	\end{equation} 
	that is, ${{{\partial ^\alpha }q} /{\partial {y^\alpha }(0)}} \ne 0$.  Similarly, we can obtain ${{{\partial ^\gamma }q} /{\partial {y^\gamma }(0)}} = 0$ for arbitrary $i$-multi-index $\gamma  \prec \alpha$. Therefore $\alpha$ is a weakly essential $i$-multi-index of $q(y)$.\par
	To prove the sufficiency, it is enough to note that the inverse transformation of $V$ is of the form
	\begin{equation}\nonumber
		\begin{aligned}
		&y_1 = U_1(x_1),\dots,y_i = U_i(x_1,\dots,x_i),\\
		&y_{i+1}=U_{i+1}(x_{i+1})=x_{i+1},\dots,y_m=U_m(x_m)=x_m 
		\end{aligned}
	\end{equation}
	and to repeat the proof of the necessity. 
\end{proof}

Furthermore, the following proposition can be verify in a similar way to the proof the above proposition.

\begin{proposition}\label{prop_ess_inv1}
	$\alpha$ is an essential multi-index of a smooth function (or a holomorphic function) $p(x_1,\dots,x_m)$ if and only if $\alpha$ is still an essential multi-index of the function $q(y) = p(V(y))$ where $x=V(y)$ is a lower transformation coordinate transformation.  
\end{proposition}

\begin{definition}\label{def_gwe}
	$I$ is a set of $i$-multi-indies and $I'$ is a subset of $I$. $I'$ is said to be the greatest weakly essential $i$-multi-index set of $I$ if $I'$ consists of all the weakly essential $i$-multi-indices of $I$. Let $p({x_1}, \dots ,{x_m})$ be a smooth function (or a holomorphic function) and $I_p$ is a subset of ${\cal I}_i(p)$. $I_p$ is said to be the greatest weakly essential $i$-multi-index set of $p$, denoted it by ${\cal W}_i(p)$, if it is the greatest weakly essential $i$-multi-index set of ${\cal I}_i(p)$. We also write ${\cal W}(p)=\bigcup_{i=0}^{m}  {\cal W}_i(p)$.
\end{definition}
\begin{definition}\label{def_ge}
	$I$ is a set of multi-indies. $I'$ is said to be the greatest essential $i$-multi-index set of $I$ if $I'$ consists of all the essential $i$-multi-indices of $I$. A set is said to be the greatest essential $i$-multi-indices set of $p({x_1}, \dots ,{x_m})$, denoted it by ${\cal E}_i(p)$, if the set consists of all the essential $i$-multi-indices of ${\cal I}(p)$. We also define ${\cal E}(p)=\bigcup_{i=0}^{m}  {\cal E}_i(p)$, and call ${\cal E}(p)$ as the greatest essential multi-indices set of $p$.
\end{definition}\par
Exploiting Definition \ref{def_gwe}, Definition \ref{def_ge}, Proposition \ref{prop_wess_inv}, and Proposition \ref{prop_ess_inv1}, we obtain the following two theorems.
\begin{theorem}
	Let $p({x_1}, \dots ,{x_m})$ be a smooth function (or a holomorphic function), $x=V(y)$ a change of coordinates taking the form \eqref{eq_tri_coord_i}, and $q(y_1,\dots,y_m)=p(V_1(y_1), \dots ,V_m(y_1,\dots,y_m))$. Then ${\cal W}_i(p) = {\cal W}_i(q)$. 
\end{theorem}
\begin{theorem}\label{theo_ess_trans}
	Let $p({x_1}, \dots ,{x_m})$ be a smooth function (or a holomorphic function), $x=V(y)$ a lower triangular coordinate transformation, and $q(y_1,\dots,y_m)=p(V_1(y_1), \dots ,V_m(y_1,\dots,y_m))$. Then ${\cal E}(p) = {\cal E}(q)$ and ${\cal E}_i(p) = {\cal E}_i(q)$ for $i=0,\dots,m$. 
\end{theorem}

\begin{proposition}
	$I$ is a set of proper $i$-multi-indices such that, for any two different elements $\alpha, \beta \in I$, both $\alpha \nprec \beta$ and $\beta \nprec \alpha$ are satisfied. Then $I$ is a finite set.
\end{proposition}
\begin{proof}
	When $i=0$, $I$ is obviously finite. Assuming $i=1$ and $\alpha = (\alpha_1) \in I$, $\alpha$ must be the only element of $I$ because, for any $\beta = (\beta_1)$ different from $\alpha$, $\beta_1<\alpha_1$ means $\beta \prec \alpha$ and $\alpha_1<\beta_1$ means $\alpha \prec \beta$. \par
	We now show that if for all $i=1,\dots,j$ the set of $i$-multi-indices $I$ is finite, then $I$ remains finite when $I$ is a set of $(j+1)$-multi-indices. Suppose $\alpha=(\alpha_1,\dots,\alpha_{j+1})$ is a given proper $(j+1)$-multi-index of $I$. For any $\beta = (\beta_1,\dots,\beta_{j+1}) \in I$, there are four possible relations between $(\alpha_1,\dots,\alpha_j)$ and $(\beta_1,\dots,\beta_j)$ as follows:  $(\alpha_1,\dots,\alpha_j)=(\beta_1,\dots,\beta_j)$, $(\alpha_1,\dots,\alpha_j) \prec (\beta_1,\dots,\beta_j)$, $(\beta_1,\dots,\beta_j) \prec (\alpha_1,\dots,\alpha_j)$, and neither $(\alpha_1,\dots,\alpha_j) \preceq (\beta_1,\dots,\beta_j)$ nor $(\beta_1,\dots,\beta_j) \preceq (\alpha_1,\dots,\alpha_j)$. We will verify that the subset consisting of all the multi-indices falling into each case is finite. In the first case, $\beta_{j+1}=\alpha_{j+1}$ must hold to meet both $\alpha \nprec \beta$ and $\beta \nprec \alpha$; that is, $\alpha$ is the only multi-index suitable for this case. In the second case,  $(\alpha_1,\dots,\alpha_j) \prec (\beta_1,\dots,\beta_j)$ means that $\beta \nprec \alpha$ has already been satisfied and we have to choose $\beta$ such that $\textstyle {\sum_{k = 1}^{j+1} {\beta _k}}  < \textstyle {\sum_{k = 1}^{j+1} {\alpha _k}}$.
	For a given $\alpha$, the above inequality implies that the choices of $\beta$ are finite. Let us discuss the third case. The number of all the proper $j$-multi-indices $(\beta_1,\dots,\beta_j)$ satisfying $(\beta_1,\dots,\beta_j) \prec (\alpha_1,\dots,\alpha_j)$ is finite. Furthermore, for a fixed $(\beta_1,\dots,\beta_j)$, there are no more than one element $\beta' \in I$ satisfying $\beta_l' = \beta_l$ for $ l=1,\dots,j$. So the elements of $I$ that meet the third case are also finite. In the last case, the two proper $j$-multi-indices $(\alpha_1,\dots,\alpha_j)$ and $(\beta_1,\dots,\beta_j)$ can not generate each other. For a given $(\alpha_1,\dots,\alpha_j)$, all the proper $j$-multi-indices that can be select as $(\beta_1,\dots,\beta_j)$ have been assumed to be finite. Note that, for a fixed $(\beta_1,\dots,\beta_j)$, at most one proper $(j+1)$-multi-index taking the form $(\beta_1,\dots,\beta_j,\beta_{j+1})$ can belong to $I$. Then, all the possible proper $(j+1)$-multi-indices that can be chosen as $\beta$ in this case are finite. In summary, the set $I$ is finite.
\end{proof}\par
The following theorem can be obtained directly from the above proposition.
\begin{theorem}\label{theo_w_fin}
	Suppose $p(x_1,\dots,x_m)$ is a smooth function (or a holomorphic function). Then, for $i=0,\dots,m$, ${\cal W}_i(p)$, ${\cal E}_i(p)$, ${\cal W}(p)$, and ${\cal E}(p)$ are all finite sets.
\end{theorem}

Let $I$ be a set of multi-indices. We write the set that consists of all the multi-indices generated by the elements of $I$ as ${\cal G}(I)$,  and write the subset that consists of all the proper $i$-multi-indices of ${\cal G}(I)$ as ${\cal G}_i(I)$.

\begin{theorem}\label{theo_min_2}
	$I$ is a set of multi-indices, and $W$ is a set of weakly essential $i$-multi-indices of $I$. Suppose $\alpha \in I \setminus {\cal G}_i(W)$ is a proper $i$-multi-index and there exists an integer $l \in \{1,\dots,i\}$ such that \par
	(i) $\textstyle{\sum_{j = 1}^l {{\alpha _j}}}  \le \textstyle{\sum_{j = 1}^l {{\beta _j}}}$
	holds for every proper $i$-multi-index $\beta \in I \setminus \left( {\cal G}_i(W) \bigcup \{\alpha\} \right)$,\par
	(ii) $\alpha \lessdot \beta$ is satisfied when $\textstyle{\sum_{j = 1}^l {{\alpha _j}}}  = \textstyle{\sum_{j = 1}^l {{\beta _j}}}$. \\
	Then $\alpha$ must be a weakly essential $i$-multi-index of $I$. Additionally, if $p(x_1,\dots,x_m)$ is a smooth function (or a holomorphic function) and the aforementioned set $I= {\cal I}(p)$, then $\alpha \in {\cal W}_i(p)$.
\end{theorem}
\begin{proof}
	Since $\beta \nprec \alpha$  and $\alpha' \nprec \alpha$ for any $\alpha' \in {\cal G}_i(W) $, $\alpha$ can not be generated by another proper $i$-multi-index of $I$. Thus, $\alpha$ is a weakly essential multi-index of $I$. 
\end{proof}

\begin{corollary}\label{coro_min_2}
	Let $I$ be a set of multi-indices and $W$ a set of weakly essential $i$-multi-indices of $I$. $\alpha$ is the least $i$-multi-indices of  $I \setminus {\cal G}_i(W)$, then $\alpha$ must be a weakly essential multi-index of $I$. Additionally, suppose $p(x_1,\dots,x_m)$ is a smooth function (or a holomorphic function) and $I$ is exactly ${\cal I}(p)$, then $\alpha \in {\cal W}_i(p)$.
\end{corollary}

	 By using Corollary \ref{coro_min_2} and Theorem \ref{theo_min_2}, the following two algorithms are provided to find the greatest weakly essential $i$-multi-index set of a set of proper $i$-multi-indices. 
	
\begin{algorithm}\label{alg_least}
	$I_i$ is a set of proper $i$-multi-indices. Determine $W_i$ the greatest weakly essential $i$-multi-index set of  $I_i$:\par
	\hangafter 0 
	\hangindent 1em 	
	Step 1) Set $W_i = \emptyset$.\par
	Step 2) If $I_i \backslash {\cal G}_i(W_i)= \emptyset$, then the algorithm terminates; else find the least $i$-multi-index of $I_i \backslash {\cal G}_i(W_i)$, denoted it by $\alpha$, set $W_i = W_i \bigcup \{ \alpha \} $, and then go to Step 2).
\end{algorithm}
	
\begin{algorithm}\label{alg_least_i}
	$I_i$ is a set of proper $i$-multi-indices. Determine $W_i$ the greatest weakly essential $i$-multi-index set of $I_i$:\par
	\hangafter 0 
	\hangindent 1em 	
	Step 1) Set $W_i = \emptyset$.\par
	Step 2) If  $I_i \setminus {\cal G}_i(W_i)= \emptyset$, then the algorithm terminates; else for every $k=1,\dots,i$ find the least multi-index of the set
	\begin{equation}\nonumber
		 \left\{ \alpha \left| \alpha  \in I_i \setminus {\cal G}_i(W_i) \, \wedge \, {\sum\limits_{j = 1}^k {{\alpha _j} = } } \right.\mathop {\min }\limits_{\alpha ' \in I_i \setminus {\cal G}_i(W_i)} \sum\limits_{j = 1}^k {{{\alpha '}_j}} \right\},
	\end{equation}
	denoted it by $\alpha^k$, set $W_i = W_i \bigcup \{ \alpha^1,\dots,\alpha^i \}$, and then go to Step 2).
\end{algorithm}\par
\begin{remark}
	For a function $p(x_1,\dots,x_m)$, the above two algorithms provide methods to obtain ${\cal W}_i(p)$ from ${\cal I}_i(p)$.
\end{remark}

	It is clear that ${\cal E}_i(p) \subseteq {\cal W}_i(p)$ for a function $p(x_1,\dots,x_m)$. In this paper, we pay special attention to ${\cal E}_m(p)$. Making use of Proposition \ref{prop_ngene}, we get the following theorem.  
\begin{theorem}\label{theo_wmem}
	Suppose $p(x_1,\dots,x_m)$ is a smooth function (or a holomorphic function). Then ${\cal E}_m(p) = {\cal W}_m(p)$.  
\end{theorem}\par
\begin{proposition}
	$p(x_1,\dots,x_m)$ is a smooth function (or a holomorphic function). Then we have
	\begin{equation}\nonumber
		{\cal E}_i(p)={{\cal W}_i(p)} \bigg \backslash { {\cal G}_i\left (\bigcup_{j=i+1}^{m} {\cal W}_j(p)\right )},
	\end{equation}
and 
	\begin{equation}\nonumber
		{\cal E}(p)= \bigcup_{j=0}^{m} {\cal E}_j(p) = {{\cal W}(p)} \bigg \backslash \bigcup_{j=0}^{m}\left ( {\cal G}\left ({\cal W}_j(p)\right)\backslash{\cal W}_j(p) \right)
	\end{equation}
\end{proposition}

\subsection{Invariant Multi-indies of Functions}
In this subsection, we consider a question that for $\alpha$ a given multi-index of a function $p({x_1}, \dots ,{x_m})$ whether there exists a lower triangular coordinate transformation $x=V(y)$ such that $\alpha$ is not a multi-index of $p(V(y_1, \dots ,y_m))$.
\begin{definition}
	$p(x_1,\dots,x_m)$ is a smooth function (or a holomorphic function). $\alpha$, a proper $i$-multi-index of function $p$ with $i \in \{0,\dots,m\}$, is said to be invariant under every lower triangular coordinate transformation $x=(x_1,\dots,x_m)=V(y)=(V_1(y_1), \dots ,V_m(y_1,\dots,y_m))$ if $\alpha$ is still a proper $i$-multi-index of the function $q(y_1, \dots ,y_m)=p(V_1(y_1), \dots ,V_m(y_1,\dots,y_m))$.
\end{definition}\par
Proposition \ref{prop_ess_inv1} implies the following proposition.
\begin{proposition}\label{prop_ess_inv} 
	All the essential multi-indices of the function $p({x_1}, \dots ,{x_m})$ are invariant. 
\end{proposition}

Now we only need to consider, for $\alpha \in {\cal I}(p) \backslash {\cal E}(p)$, whether there exists a lower triangular coordinate transformation $x=V(y)$ such that $\alpha$ is not a multi-index of $q(y)=p(V(y))$. The next example illustrates that this kind of lower triangular coordinate transformation may not exist when we restrict it to real-value coordinate transformations.

\begin{example}
	Consider the function
	\begin{equation}\nonumber
		{p_{\rm{1}}}(x_1,x_2) = x_1 x_2^2 - x_1^{\rm{3}}.
	\end{equation}
${p_{\rm{1}}}(x_1,x_2)$ has proper $2$-multi-indices $(1,2)$ and $(3,0)$. $(1,2)$ is the least $2$-multi-index of $p_1$ and can generate $(3,0)$. Select a lower triangular coordinate transformation $y=U(x)$ as
	\begin{equation}\nonumber
		\begin{aligned}
			&{y_1} = {x_1}\\
			&{y_2} =  - {x_1} + {x_2},
		\end{aligned}
	\end{equation}
the inverse  transformation of which, denoted by $x=V(y)$, is
	\begin{equation}\nonumber
		\begin{aligned}
			&{x_1} = {y_1}\\
			&{x_2} = {y_1} + {y_2}.
		\end{aligned}
	\end{equation}
We rewrite $p_1$ in $y$-coordinates 
	\begin{equation}\nonumber
		{p_1}(V(y_1,y_2))  = {y_1}y_2^2 + 2y_1^2{y_2}.
	\end{equation}	
$(3,0)$ is not a multi-index of $p_1(V(y))$. Now consider another function
	\begin{equation}\nonumber
		{p_{\rm{2}}}(x_1,x_2) = x_1x_2^2{\rm{ + }}x_1^{\rm{3}}.
	\end{equation}	
Choose a lower triangular coordinate transformation 
	\begin{equation}\label{eq_tri_coord_2}
	\begin{aligned}
		&{x_1} = {d_{11}}{y_1} + {r_1}(y_1)\\
		&{x_2} = {d_{21}}{y_1} + {d_{22}}{y_2} + {r_2}(y_1,y_2)
	\end{aligned}
	\end{equation}
where $d_{11}, d_{21}, d_{22}$ are parameters with $d_{11},d_{22} \ne 0$ and $r_1, r_2$ are smooth functions with $\partial r_1/\partial y_1(0) = 0$, $\partial r_2/\partial y_1(0) = 0$,  and $\partial r_2/\partial y_2(0) = 0$. In $y$-coordinates, we have
	\begin{equation}\label{eq_p2}
	\begin{aligned}
		p_2(V(y_1,y_2)) =& d_{11}d_{22}^2y_1y_2^2 + 2d_{11}{d_{21}}{d_{22}}y_1^2{y_2} \\
		&+ {d_{11}}(d_{11}^2 + d_{21}^2)y_1^{\rm{3}} + \dots\;,
	\end{aligned}
	\end{equation}
	where we only present all the cubic terms of $p_2(V(y_1,y_2))$. Because of the arbitrariness of \eqref{eq_tri_coord_2}, it is impossible to find a real-valued smooth lower triangular coordinate transformation such that ${d_{11}}(d_{11}^2 + d_{21}^2) = 0 $. In order to eliminate the multi-index $(3,0)$ from the right-hand side of \eqref{eq_p2}, we have to take $d_{11}$ and $d_{21}$ as complex numbers.\par
\end{example}

	The above example prompts us to use complex-valued lower triangular coordinate transformations.
	\begin{theorem}
		Let $p({x_1}, \dots ,{x_m})$ be a smooth function (or a holomorphic function). A multi-index of $p$ is invariant under any biholomorphic lower triangular coordinate transformations if and only if it belongs to ${\cal E}(p)$.
	\end{theorem}
	\begin{proof}
		We only prove that, for $h=1,\dots,m$ and a proper $h$-multi-index $\alpha \in {\cal I}(p) \backslash {\cal E}(p)$, there exists a biholomorphic lower triangular coordinate transformation  $x=V(y)$ such that $\alpha \notin {\cal I}(q)$ for $q(y) = p(V(y))$. \par
		Let $r$ be a positive integer and $\alpha^1,\dots,\alpha^r \in {\cal I}(p)$ be all the multi-indices each of which can generate $\alpha$ and is different from $\alpha$. Let
		\begin{equation}\nonumber
			{\lambda ^i} = (0, \dots ,0,1)
		\end{equation} 
	be proper $i$-multi-index for $i=1,\dots,m$. Choose a biholomorphic lower triangular coordinate transformation $x=V(y)$ as 
	\begin{equation}\label{eq_inv_coord}
		\begin{aligned}
			&{x_1} = V_1(y_1) ={c_{{\lambda ^1}}}{y_1} + \sum\limits_{j=1}^{j_1} {{c_{{\beta ^{1,j}}}}{y^{{\beta ^{1,j}}}}} \\
			&\quad \vdots \\
			&{x_m} = V_m(y_1,\dots,y_m) = {c_{{\lambda ^m}}}{y_m} + \sum\limits_{j=1}^{j_m} {{c_{{\beta ^{m,j}}}}{y^{{\beta ^{m,j}}}}} 
		\end{aligned}
	\end{equation}	
	where ${c_{{\lambda ^1}}}, \dots ,{c_{{\lambda ^m}}} \ne 0$ are given real numbers, and ${c_{{\beta ^{i,j}}}}$, $i = 1,\dots,m$, $j =1,\dots,j_i$, and $j_i \ge 0$, are undetermined complex-valued  coefficients.
	The multi-indices 
	\begin{equation}\label{eq_beta}
	\begin{aligned}
		&{\beta ^{1,1}}, \dots ,{\beta ^{1,{j_1}}},\\
		&\quad \vdots \\
		& {\beta ^{m,1}}, \dots ,{\beta ^{i,{j_m}}},
	\end{aligned}
	\end{equation}
	introduced in \eqref{eq_inv_coord} satisfy three conditions: \par
	\begin{enumerate}
		\item $\beta^{i,j}$, $i=1,\dots,m$, $j=1,\dots,j_i$, and $j_i \ge 0$, are $i$-multi-indices with $\beta^{i,j} \ne 0$ and $\beta^{i,j} \ne \lambda^i$. 
		\item There exist at least one multi-index $\alpha^k=(\alpha_1^k,\dots,\alpha_m^k)$ with $k \in \{1,\dots,r\}$ and a family of multi-indices $\gamma ^{i,j}$ ($i=1,\dots,m$ and $j=1,\dots,\alpha^k_i$) selected from $\lambda^i,\beta^{i,1},\dots,\beta^{i,j_i}$ such that
		\begin{equation}\label{eq_gen_alpha}
			\alpha =\sum\limits_{i=1}^{m} {\;\sum\limits_{j=1}^{\alpha_i^k} {{\gamma ^{i,j}}} }  =
			\sum\limits_{i'=1}^{m}n^{i',0}\lambda^{i'} + \sum\limits_{i'=1}^{m} {\sum\limits_{j'=1}^{j_{i'}} {n^{i',j'} \beta ^{i',j'}} } 
		\end{equation}
		where all the $n^{i',0}$ are nonnegative integers, all the $n^{i',j'}$ ($j' \ge 1$) are positive integers, and $\sum_{j'=0}^{j_{i'}} n^{i',j'} =\alpha^k_{i'}$. 
		\item if any multi-index listed in \eqref{eq_beta} is removed, \eqref{eq_gen_alpha} is not satisfied for all $\alpha^1,\dots,\alpha^r$.	
	\end{enumerate}
	The existence of \eqref{eq_beta} is guaranteed by $\alpha^k \prec \alpha$ for $k=1,\dots,r$. Without loss of generality, assume that $\alpha^1,\dots,\alpha^s$ with $1 \le s \le r$ satisfy \eqref{eq_gen_alpha}. $p$ can be expressed as 
	\begin{equation}\nonumber
	\begin{aligned}
		p({x_1}, \dots ,{x_m}) =&{c_\alpha }{x^\alpha } + {c_{{\alpha ^1}}}{x^{{\alpha ^1}}} +  \dots  + {c_{{\alpha ^r}}}{x^{{\alpha ^s}}} \\
		&+ p'({x_1}, \dots ,{x_m})
	\end{aligned}
	\end{equation}
	where $c_\alpha, c_{{\alpha ^1}}, \dots, c_{{\alpha ^s}}$ are nonzero coefficients, and  $\alpha, \alpha ^1, \dots, \alpha ^s$ are not multi-indices of function $p'({x_1}, \dots ,{x_m})$. We also assume that for a fixed $\alpha^k$ there are different $t_k \ge 1$ families of integers $n^{1,0},\dots,n^{1,j_1},\dots\dots,n^{m,0},\dots,n^{m,j_m}$ satisfying \eqref{eq_gen_alpha}. Substituting \eqref{eq_inv_coord} into $p(x_1,\dots,x_m)$ and taking account of the requirement that $\alpha$ should be not a multi-index of $q(y_1,\dots,y_m)=p(V(y))$ yield
	\begin{equation}\label{eq_coef}
		{c_\alpha }\prod\limits_{i = 1}^h {c_{{\lambda ^i}}^{{\alpha _i}}}  + \sum\limits_{k = 1}^s  \sum\limits_{l = 1}^{t_k}    {\left( {{c_{{\alpha ^k}}}\prod\limits_{i = 1}^m {\prod\limits_{j = 1}^{\alpha _i^k} {{c_{{\chi ^{i,j,k,l}}}}} } } \right)}  = 0
	\end{equation}
	where every $\chi ^{i,j,k,l}$ is selected from $\lambda^i,\beta ^{i,1}$,$\dots$,$\beta ^{i,j_i}$,  correspondingly ${c_{{\chi ^{i,j,k,l}}}}$ is selected from $c_{\lambda^i},{c_{{\beta ^{i,1}}}}$,$\dots$,${c_{{\beta ^{i,j_i}}}}$, and $\sum_{i,j} {{\chi ^{i,j,k',l'}}}  = \alpha $ holds for any fixed pair  of the numbers $k'$ and $l'$. From condition 3), all the undetermined coefficients ${c_{{\beta ^{i,j}}}}$, $i = 1,\dots,m$ and $j =1,\dots,j_i$, are factors of every term in the left-hand side of \eqref{eq_coef} except ${c_\alpha }\prod_{i = 1}^h {c_{{\lambda ^i}}^{{\alpha _i}}}$.  
	
	It remains to verify that there exist ${c_{{\beta ^{i,j}}}}$, $i = 1,\dots,m$ and $j =1,\dots,j_i$, such that \eqref{eq_coef} holds. For convenience, rename ${\beta ^{1,1}}, \dots ,{\beta ^{1,{j_1}}}, \dots  \dots, {\beta ^{m,1}}, \dots ,{\beta ^{m,{j_m}}}$ to $\beta^1,\dots,\beta^{j_1+\dots+j_m}$, and rename $c_{\beta ^{1,1}}, \dots ,c_{\beta ^{1,{j_1}}}$, $\dots  \dots$, $c_{\beta ^{m,1}}, \dots ,c_{\beta ^{m,{j_m}}}$ to $c_{\beta^1},\dots,c_{\beta^{j_1+\dots+j_m}}$. Let us regard the left-hand side of \eqref{eq_coef} as a polynomial in indeterminate $c_{\beta ^{1}}$, denoted the polynomial by $P_1(c_{\beta ^{1}})$, and assume that the degree of $P_1(c_{\beta ^{1}})$ is $e_{1}$. Then, the polynomial can be rewritten in the form 	
	\begin{equation}\nonumber
		P_1(c_{\beta ^{1}}) = P_{2} c_{\beta ^{1}}^{e_{1}}+R_{1} 
	\end{equation}
	where $P_{2}$ and $R_{1}$ are functions satisfying $\partial P_{2}/ \partial c_{\beta ^{1}} = 0$ and $\partial^{e_{1}} R_{1}/ \partial c_{\beta ^{1}}^{e_{1}} = 0$. $P_{2}$ can also be regarded as a polynomial in indeterminate $c_{\beta ^{2}}$. Let us, in general, consider $P_k$ ($k = 1,\dots,j_1+\dots+j_m$) as a polynomial in indeterminate $c_{\beta ^{k}}$ and suppose the degree of $P_{k}(c_{\beta ^{k}})$ is $e_{k}$ ($e_{k} \ge 1$), then we have 
	\begin{equation}\label{eq_poly1}
		P_{k}(c_{\beta ^k}) = P_{k+1} c_{\beta ^{k}}^{e_{k}}+R_{k} 
	\end{equation}
	where $P_{k+1}$ is a function satisfying $\partial P_{k+1}/ \partial c_{\beta ^{\bar k}} = 0$ for $\bar k =1,\dots,k$, and $R_{k}$ is a function satisfying $\partial^{e_k} R_{k}/ \partial c_{\beta ^k}^{e_k} = 0$, $\partial R_{k}/ \partial c_{\beta ^{\hat k}} = 0$ for $\hat k =1,\dots,k-1$, and $c_{\beta ^k}$ is a factor of every term in $R_{k}$. It is clear that $P_{k+1} $ can be regarded as a polynomial in indeterminate $c_{\beta ^{k+1}}$ if $k+1 \le j_1+\dots+j_m$ is satisfied. Since any two of the multi-indices $\alpha^1,\dots,\alpha^s$ are different from each other, we know that $P_{j_1+\dots+j_m+1}$ must be a nonzero constant. Setting $P_{j_1+\dots+j_m}=r_{j_1+\dots+j_m} \ne R_{j_1+\dots+j_m}(0)=0$ where $r_{j_1+\dots+j_m}$ is a constant, \eqref{eq_poly1} has at least one nonzero solution for $c_{\beta ^{j_1+\dots+j_m}}$. When $c_{\beta ^{j_1+\dots+j_m}},\dots,c_{\beta ^{k+1}}$ have been determined for $k =j_1+\dots+j_m-1,\dots,2$, we set $P_{k}=r_k \ne R_k(0,c_{\beta ^{k+1}},\dots,c_{\beta ^{j_1+\dots+j_m}})=0$ where $r_k$ is a constant, and then we can find a nonzero $c_{\beta ^{k}}$ satisfying \eqref{eq_poly1}. We finally solve \eqref{eq_coef} for a nonzero $c_{\beta ^{1}}$. Therefore, an appropriate lower triangular coordinate transformation $x=V(y)$ such that $\alpha \notin {\cal I}(q)$ is obtained.
	\end{proof}

\subsection{Classifications of Lower Triangular Forms}
Having finished the previous discussions about the invariant multi-indices of functions, let us investigate what properties of lower triangular forms are invariant under lower triangular coordinate transformations.
\begin{definition}
$\alpha = (\alpha_1,\dots,\alpha_j)$ and $\beta = (\beta_1,\dots,\beta_j)$ are multi-indices. We write $\alpha \le \beta$ if $\alpha_i \le \beta_i$ holds for all $i=1,\dots,j$, and write $\alpha < \beta$ if $\alpha \le \beta$ and $\alpha \ne \beta$ \cite{rudin1991functional}.
\end{definition}\par
\begin{remark}
	Suppose $\alpha$ and $\beta$ are proper $j$-multi-indices. $\alpha \le \beta$ implies $\alpha \preceq \beta$.
\end{remark}
\begin{proposition}
	$p(x_1,\dots,x_{m})$ and $q(x_1,\dots,x_{m-1})$ are smooth functions with $p(0)=0$ and $q(0) \ne 0$. Then ${\cal E}_m(p) = {\cal E}_m(p \cdot q)$.
\end{proposition}
\begin{proof}
	Let ${q_1}({x_1}, \dots ,{x_m-1}) = q({x_1}, \dots ,{x_m-1}) - q(0)$. The following equation is the well-known Leibniz formula \cite{rudin1991functional}
	\begin{equation}\label{eq_Leibniz}
		\frac{{\partial ^\alpha }(p \cdot {q_1})}{\partial x^\alpha} = \sum\limits_{\beta  \le \alpha } \left( \frac{{\prod\limits_{i = 1}^m {{\alpha _i}!} }}{{\prod\limits_{i = 1}^m {({\beta _i}!({\alpha _i} - {\beta _i})!)} }} \frac{{\partial ^\beta }p}{\partial x^\beta} \frac{{\partial ^{\alpha  - \beta }}{q_1}}{\partial x^{\alpha  - \beta }} \right)
	\end{equation}
	where $\alpha$ is proper $m$-multi-index.
	Assuming $\gamma \in {\cal E}_m(p)$ and $\alpha \preceq \gamma$, \eqref{eq_Leibniz} yields ${{\partial ^\alpha }(p \cdot {q_1})}/{\partial x^\alpha}(0) = 0$; that is, $\alpha$ is  not a multi-index of $p \cdot q_1$. On the other hand, ${\cal E}_m(q(0) \cdot p) = {\cal E}_m(p)$. Thus we have ${\cal E}_m(p) \subseteq {\cal E}_m(p \cdot q)$. Now let $\alpha \notin {\cal E}_m(p) $ be a proper $m$-multi-index satisfying both $\alpha \npreceq \gamma$ and $\gamma \npreceq \alpha$  for all $\gamma \in {\cal E}_m(p)$, and let $\beta  \le \alpha$ be a multi-index. Since ${\partial ^\beta p}/{\partial x^\beta}(0) = 0$ holds if $\beta$ is a proper $m$-multi-index and ${\partial ^{\alpha-\beta} q_1}/{\partial x^{\alpha-\beta}}(0) = 0$ holds if $\beta$ is not a proper $m$-multi-index, we have ${\partial ^{\alpha}}(p \cdot {q_1})(0) = 0$, which implies ${\cal E}_m(p \cdot q) \subseteq {\cal E}_m(p)$. In conclusion, ${\cal E}_m(p) = {\cal E}_m(p \cdot q)$.   
\end{proof}
\begin{theorem}\label{theo_indices}
	Suppose $y=U(x)$ is a lower triangular coordinate transformation, and rewrite \eqref{eq_tri} in $y$-coordinates as follows
	\begin{equation}\label{eq_tri_y}
		\begin{aligned}
			&{{\dot y}_1} = {{\bar f}_1}({y_1},{y_2})\\
			&\; \vdots \\
			&{{\dot y}_{n - 1}} = {{\bar f}_{n - 1}}({y_1} \dots ,{y_n})\\
			&{{\dot y}_n} = {{\bar f}_n}({y_1}, \dots ,{y_n}) + {{\bar g}_n}({y_1}, \dots ,{y_n})v\;.
		\end{aligned}
	\end{equation}
	Then ${\cal E}_{i+1}({f_i}) = {\cal E}_{i+1}({\bar f_i})$ holds for any $i = 1, \dots ,n - 1$.
\end{theorem}
\begin{proof}
	Let us compute the ${{\bar f}_i}({y_1} \dots ,{y_{i+1}})$ in $x$-coordinates
	\begin{equation}\nonumber
	\begin{aligned}
			&{{\bar f}_i}({y_1}, \dots ,{y_{i + 1}}) = \frac{{\partial {U_i}}}{{\partial {x_i}}}{f_i}({x_1}, \dots ,{x_{i + 1}})\\
			&\qquad + \sum\limits_{k = 1}^{i - 1} {\frac{{\partial {U_i}}}{{\partial {x_k}}}} {f_i}({x_1}, \dots ,{x_{k+1}}) = {\bar f}_i(U({x_1}, \dots ,{x_{i + 1}}))\;.
	\end{aligned}
	\end{equation}
	Thanks to the above proposition, we have ${\cal E}_{i+1} ({{\partial {U_i}}}/{{\partial {x_i}}} \cdot {f_i}) = {\cal E}_{i+1}({f_i})$. In addition, ${\cal E}_{i+1} ({{\partial {U_i}}}/{{\partial {x_k}}} \cdot {f_k}) = \emptyset $ is satisfied for any $k =1,\dots,i-1$. Therefore ${\cal E}_{i+1} ({\bar f}_i(U({x_1}, \dots ,{x_{i + 1}}))) = {\cal E}_{i+1}({f_i})$ holds. Using Theorem \ref{theo_ess_trans}, we conclude ${\cal E}_{i+1} ({\bar f}_i({y_1}, \dots ,{y_{i + 1}})) = {\cal E}_{i+1}({f_i})$. 
\end{proof}\par
\begin{corollary}
	Suppose $y=U(x)$ is a lower triangular coordinate transformation.  Rewriting \eqref{eq_tri} in $y$-coordinates yields \eqref{eq_tri_y}. Then ${\cal L}_{i+1}({f_i}) = {\cal L}_{i+1}({\bar f_i})$ for $i = 1, \dots ,n - 1$.
\end{corollary}\par
This corollary leads to a way to classify lower triangular forms. 
\begin{definition}\label{def_typeL}
	All the lower triangular forms taking the form \eqref{eq_tri} and satisfying ${\cal L}_{i+1}(f_i) =\alpha^i$ for $i=1,\dots,n-1$ are grouped under a specific type, denoted by $[\alpha^1,\dots,\alpha^{n-1}]$. Arbitrary element of $[\alpha^1,\dots,\alpha^{n-1}]$ can be expressed as
	\begin{equation}
		\begin{aligned}\label{eq_lea_sys}
			&{{\dot x}_1} = {c_1}{x^{{\alpha ^1}}} + {{\hat f}_1}(x_1,x_2)\\
			&\; \vdots \\
			&{{\dot x}_{n - 1}} = {c_{n - 1}}{x^{{\alpha ^{n - 1}}}} + {{\hat f}_{n - 1}}(x_1 \dots ,x_n)\\
			&{{\dot x}_n} = {f_n}({x_1}, \dots ,{x_n}) + {g_n}({x_1}, \dots ,{x_n})v
		\end{aligned}
	\end{equation} 
	where, for any $i=1,\dots,n-1$, ${\hat f}_i$ is smooth function vanishing at the origin and $\alpha^i \lessdot \beta$, provided that $\beta$ is any $(i+1)$-multi-index of ${\hat f}_i$, is satisfied.
\end{definition}\par  
\begin{remark}
	System \eqref{eq_p} is of type
	\begin{equation}\nonumber
		[(0,p_1), (0,0,p_2), \dots, (0,\dots,0,p_{n-1})]\;.
	\end{equation} 
\end{remark}

Theorem \ref{theo_indices} results in another way to classify lower triangular forms. 
\begin{definition}\label{def_typeE}
	All the systems taking the form \eqref{eq_tri} and having the same ${\cal E}_{i+1}(f_i)$ for $i=1,\dots,n-1$ can be expressed as
	\begin{equation}\label{eq_ess_sys} 
		\begin{aligned}
			&{{\dot x}_1} =f(x_1,x_2)= \sum\limits_{\alpha  \in {\kern 1pt} {{\cal E}_2(f_1)}} {c_1^\alpha {x^\alpha }}  + {{\tilde f}_1}({x_1,x_2})\\
			&\; \vdots \\
			&{{\dot x}_{n - 1}} = {f_{n-1}}({x_1}, \dots ,{x_n})\\
			&\;\;= \sum\limits_{\alpha  \in {\kern 1pt} {{\cal E}_n(f_{n-1})}} {c_{n - 1}^\alpha {x^\alpha }}  + {{\tilde f}_{n - 1}}({x_1} \dots ,{x_{n}})\\
			&{{\dot x}_n} = {f_n}({x_1}, \dots ,{x_n}) + {g_n}({x_1}, \dots ,{x_n})v
		\end{aligned}
	\end{equation} 
	where every ${\tilde f}_i$ for $i=1,\dots,n-1$ is smooth functions with ${\tilde f}_i(0)=0$ and every proper $(i+1)$-multi-index of ${\tilde f}_i$ can be generated by some element of ${\cal E}_{i+1}(f_i)$. For sake of convenience, we say that \eqref{eq_ess_sys} is of type $[\kern-0.15em[ {\cal E}_2(f_1) , \dots ,{\cal E}_n(f_{n-1})  
	]\kern-0.15em]$.     
	
\end{definition}\par
\begin{remark}\label{rem_controlterms}
	Apart from the invariance of ${\cal L}_{i+1}(f_i)$ and ${\cal E}_{i+1}(f_i)$ under lower triangular coordinate transformations, another reason we think the two classifications given in Definition \ref{def_typeL} and \ref{def_typeE} are helpful is as follows. For a lower triangular form taking the form \eqref{eq_tri}, $x_{i+1}$ can be seen as a control input of $\dot x_i = f_i(x_1,\dots, x_{i+1})$ to some extent, such as designing a feedback controller for  \eqref{eq_tri} using backstepping. So we may also consider $(x_1,\dots,x_{i+1})^{\alpha}$ where $\alpha = {\cal L}_{i+1}(f_i)$ or $\alpha \in {\cal E}_{i+1}(f_i)$ as one of the "control" terms for $\dot x_i = f_i(x_1,\dots, x_{i+1})$. From some literature, such as \cite{qian2001continuous,lin2002adaptive,hong2006adaptive,long2015integral}, we know that, at least for several types of lower triangular forms, there are some control strategies that can be applied to the entire type of lower triangular form to meet some control objectives, no matter what $\hat{f}_i$ and $\tilde{f}_i$ are. Of course, for many other types of lower triangular forms, a control strategy may only be effective when $\hat{f}_i$ and $\tilde{f}_i$ satisfy certain conditions, such as \cite{lin2000adding,qian2001non,qian2002practical,qian2006recursive,sun2015new,long2012h,su2014stabilization,ding2020nonsingular,ding2021low}. We look forward to more research on the control algorithms for \eqref{eq_lea_sys} and \eqref{eq_ess_sys}.
\end{remark}
\begin{remark}
	If the proper $(i+1)$-multi-index $(0,\dots,0,1)$ belongs to ${\cal E}_{i+1}(f_i)$ then it is the only element of ${\cal E}_{i+1}(f_i)$. In addition, the proper $(i+1)$-multi-index $(0,\dots,0,k)$ with $k \ge 1$ can generate any proper $(i+1)$-multi-index $\alpha$ satisfying $\left | \alpha \right | \ge k$. So there are at most a finite number of proper $(i+1)$-multi-indices that can not be generated by ${\cal E}_{i+1}(f_i)$ when $(0,\dots,0,k) \in {\cal E}_{i+1}(f_i)$; see the following proposition.
\end{remark}

\begin{proposition}\label{prop_finiteset}
	$I$ is a set of proper $i$-multi-indices and ${\cal A}_i$ represents the set consisting of all the proper $i$-multi-indices. ${\cal A}_i \setminus {\cal G}_i(I)$ is finite if and only if one can find some positive integer $k$ for which the proper $i$-multi-index $\lambda^{i,k}=(0,\dots,0,k)$ belongs to $I$. 
\end{proposition}
\begin{proof}
	The sufficiency is obvious, we only prove the necessity. Assume $\lambda^{i,k} \not\in I$ for all positive integer $k$. $\alpha$ is arbitrary element of ${\cal E}_i(I)$. With the assumption in mind, $\lambda^{i,k} \not\in {\cal G}_i(\{\alpha\})$ for all $k>0$ because all the proper $i$-multi-indices that can generate $\lambda^{i,k}$ are $\lambda^{i,k'}$, $k'=1,\dots,k$. It follows that $\lambda^{i,k} \not\in {\cal G}_i({\cal E}_i(I))={\cal G}_i(I)$ for all $k>0$. This means that ${\cal A}_i \setminus {\cal G}_i(I)$ is infinite. This contradiction completes the proof.  
\end{proof}

\begin{example} 
	Consider the following lower triangular form.
	\begin{equation}\nonumber
		\begin{aligned}
			&{{\dot x}_1} = \sin x_2^3 + {x_1}{x_2}\\
			&{{\dot x}_2} = {x_3}x_2^3 + {x_3}x_1^3 + {x_2}\\
			&{{\dot x}_3} = x_4^3{x_3} + {x_4}{x_1} + {x_3}\\
			&{{\dot x}_4} = {x_4} + v
		\end{aligned}
	\end{equation} 
	Let us focus on the functions expressed by the right-hand sides of the first three equations of the above system. From the least multi-indices of those functions, this system is of type
	\begin{equation}\nonumber
		[(0,3),(0,3,1),(0,0,1,3)],
	\end{equation}
	and, after having computed essential multi-indices of those functions, we know that the system is also of type 	
	\begin{equation}\nonumber 
		\left[\kern-0.15em\left[ {\{ (0,3),(1,1)\} ,\{ (0,3,1)\} ,\{ (0,0,1,3),(1,0,0,1)\} } 
		\right]\kern-0.15em\right].
	\end{equation}
\end{example}

\section{Feedback Equivalence}
In this section, we solve the problem of whether a nonlinear system is feedback equivalent to a given type of lower triangular form in two methods. The first one helps us determine what types the system belongs to by transforming the system into a lower triangular form if it is possible. And when the second method is adopted, we solve the problem by calculating a series of Lie brackets. 
\subsection{Transforming into Lower Triangular Forms}
Using the notation of the differential geometry, we write the drift vector field and the input vector field of \eqref{eq_s} as
\begin{equation}\nonumber
	F = {F_1}\frac{\partial }{{\partial {\xi _1}}}+  \dots  + {F_n}\frac{\partial }{{\partial {\xi _n}}},
\end{equation} 
and
\begin{equation}\nonumber
	G = {G_1}\frac{\partial }{{\partial {\xi _1}}} +  \dots  + {G_n}\frac{\partial }{{\partial {\xi _n}}},
\end{equation} 
respectively. Similarly, the drift vector field and the input vector field of \eqref{eq_tri} can be denoted by
\begin{equation}\nonumber
	f = {f_1}\frac{\partial }{{\partial {x_1}}} +  \dots  + {f_n}\frac{\partial }{{\partial {x_n}}},
\end{equation} 
and
\begin{equation}\nonumber
	g = {g_n}\frac{\partial }{{\partial {x_n}}}.
\end{equation}
Let $X$ and $Y$ be two $n$ dimensional vector fields defined on a neighborhood of the origin, $\mathrm{ad}_{X}Y=[X,Y]$ is the Lie bracket of $X$ and $Y$. Further let $h(\xi_1,\dots,\xi_n)$ be a smooth function, then $X(h)= \sum_{i=1}^{n} X_i \cdot \partial h/\partial \xi_i$.\par 
Though the sufficient and necessary condition under which a nonlinear system is equivalent to a lower triangular form has been already given in \cite{vcelikovsky1996equivalence}, let us show here a new condition that may be easier to check and may simplify the implementation of the equivalent transformation.
\begin{theorem}\label{theorem_tri}
	Let $D ^{n + 1} =  {\rm{span}}\left\{ 0 \right\}$, $D^n = \hat{D^n}  = {\rm{span}}\left\{ G \right\}$. System \eqref{eq_s} is locally equivalent to \eqref{eq_tri} via a feedback \eqref{eq_fb} and a change of coordinates \eqref{eq_ct} if and only if, for every $i=n-1,\dots,1$, \eqref{eq_s} satisfies the following condition: suppose $D^{k}$, $k=n+1,\dots,i+1$, and $\hat D^{l}$, $l=n,\dots,i+1$, have already been defined, take a vector field $G^{i+1} \in D^{i+1} \setminus D^{i+2}$, and set ${\hat D^{i}}= \mathrm{span} \{ \mathrm{ad}_{G^{i+1}}F,  \hat D^{i+1} \}$, then there exists an $n-i+1$ dimensional involutive distribution $D^i$ in a neighborhood of the origin such that  $D^i = \hat D^i$ in an open subset of ${{\mathbb{R}}^n}$ whose closure is a neighborhood of the origin.
\end{theorem}
\begin{proof}
	Note that, for any smooth vector fields $X=\textstyle \sum_{i=1}^n X_i {\partial}/{\partial x_i}$ and $Y=\textstyle \sum_{i=1}^n Y_i {\partial}/{\partial x_i}$, we have
	\begin{equation}\label{eq_pushforward}
		S_*([X,Y]) = [S_*(X),S_*(Y)]
	\end{equation}
	where $S:\mathbb{R}^n \to \mathbb{R}^n, x \mapsto \xi$ is a change of coordinates and $S_*$ is so-called the differential of $S$ or the pushforward induced by $S$ \cite{lang2012differential}. By using \eqref{eq_pushforward}, the necessity is clear because \eqref{eq_tri} satisfies the condition given in the theorem. Let us verify that the condition is sufficient. Due to ${D^n} \subset  \dots   \subset {D^1}$, we can find a change of coordinates $x = T(\xi )$ such that
		${D^i} = {\rm{span}}\left\{ {\partial }/{{\partial {x_i}}, \dots ,\partial }/{{\partial {x_n}}} \right\}$ \cite{lang2012differential}. 
	Let ${g^i} = {T_*}({G^i})$ and $f = {T_*}(F)$. Since  ${g^i} \in {D^i} $ for $i=n,\dots,1$, it can be expressed as ${g^i} = \textstyle \sum_{k = i}^n {g_k^i {\partial }/{{\partial {x_k}}}} $. Calculate the vector field ${\rm{a}}{{\rm{d}}_{{g^i}}}f$ as follows.
	\begin{equation}\nonumber
	\begin{aligned}
		{\rm{a}}{{\rm{d}}_{{g^i}}}f =&\sum\limits_{k = 1}^{i - 1} {\left( {\sum\limits_{j = i}^n {\frac{{\partial {f_k}}}{{\partial {x_j}}}g_j^i} } \right)\frac{\partial }{{\partial {x_k}}}}\\ 
		& + \sum\limits_{k = i}^n {\left( {\sum\limits_{j = i}^n {\frac{{\partial {f_k}}}{{\partial {x_j}}}g_j^i - } \sum\limits_{j = 1}^n {\frac{{\partial g_k^i}}{{\partial {x_j}}}{f_j}} } \right)} \frac{\partial }{{\partial {x_k}}}
	\end{aligned}
	\end{equation}
		Then, $\mathrm{ad}_{g^i}f \in {D^{i - 1}} \setminus {D^i}$, $i=n,\dots,2$, result in  
	\begin{equation}\nonumber
		\frac{{\partial {f_{i - 1}}}}{{\partial {x_i}}} \not\equiv 0,\;
		\frac{{\partial {f_j}}}{{\partial {x_i}}} \equiv 0,j = 1,\dots ,i - 2
	\end{equation}
	in a neighborhood of the origin. Thus \eqref{eq_s} in $x$-coordinates is of the form \eqref{eq_tri}.	
\end{proof}\par
\begin{remark}
	It is also clear that if a nonlinear system satisfies the condition given in Theorem \ref{theorem_tri} then the system can be transformed into a lower triangular form only via a change of coordinates.
\end{remark}
\begin{remark}
	Taking $G^n=G$ and $G^i=\mathrm{ad}_{G^{i+1}}F$ for $i=n-1,\dots,2$, the condition introduced in the above theorem is the same as the condition presented in \cite{vcelikovsky1996equivalence}. By choosing appropriate $G^i$, the calculations of the Lie brackets and design of equivalent transformation can be simplified. 
\end{remark}

The next example shows how to transform a system into its equivalent lower triangular form by using Theorem \ref{theorem_tri} and determine the types the system belongs to. 
\begin{example}\label{exa_transformation}
	Let us consider a nonlinear system expressed by \eqref{eq_nse1}
	\begin{figure*}
		\begin{equation}\label{eq_nse1}
			\begin{aligned}
				\dot{\xi}_{1}=& \xi_{1}-\xi_{3}+\left(\xi_{1}-\xi_{3}\right)\left(\xi_{2}-\xi_{3}^{2}+\xi_{4}\right)+\left(\xi_{1}-\xi_{3}+\xi_{4}\right)\left(\xi_{2}-\xi_{3}^{2}+\xi_{4}\right)+\left(\xi_{2}-\xi_{3}^{2}+\xi_{4}\right)^{3} \\
				\dot{\xi}_{2}=& \xi_{2}-\xi_{3}^{2}+2 \xi_{3}\left(\xi_{1}-\xi_{3}+\left(\xi_{1}-\xi_{3}+\xi_{4}\right)\left(\xi_{2}-\xi_{3}^{2}+\xi_{4}\right)\right)+\xi_{3}+\xi_{4} \\
				&+\left(\xi_{1}-\xi_{3}\right)\left(\xi_{2}-\xi_{3}^{2}+\xi_{4}\right)+\left(\xi_{2}-\xi_{3}^{2}+\xi_{4}\right)^{3}-\left(\left(\xi_{1}-\xi_{3}+\xi_{4}\right)^{2}+1\right) u \\
				\dot{\xi}_{3}=& \xi_{1}-\xi_{3}+\left(\xi_{1}-\xi_{3}+\xi_{4}\right)\left(\xi_{2}-\xi_{3}^{2}+\xi_{4}\right) \\
				\dot{\xi}_{4}=&\left(-\xi_{1}+\xi_{3}-\left(\xi_{2}-\xi_{3}^{2}+\xi_{4}\right)^{2}\right)\left(\xi_{2}-\xi_{3}^{2}+\xi_{4}\right)+\left(\left(\xi_{1}-\xi_{3}+\xi_{4}\right)^{2}+1\right) u
			\end{aligned}
		\end{equation}
		\hrulefill
	\end{figure*}
	and denote the drift vector field and input vector field of the system by $F(\xi)$ and $G(\xi)$.	
	Select a nonsingular vector field $G^4(\xi) \in D^4 = {\rm{span}} \{ G(\xi)\}$ as
	\begin{equation}\nonumber
		G^4(\xi)=G(\xi) \bigg{/}\left(\left(\xi_{1}-\xi_{3}+\xi_{4}\right)^{2}+1\right)=-\frac{\partial}{\partial \xi_{2}}+\frac{\partial}{\partial \xi_{4}}
	\end{equation}
	and calculate the Lie bracket of $G^4$ and $F$
	\begin{equation}\nonumber
		\begin{aligned}
			\mathrm{ad}_{G^4} F=&\left(\xi_{2}-\xi_{3}^{2}+\xi_{4}\right) \frac{\partial}{\partial \xi_{1}}+2 \xi_{3}\left(\xi_{2}-\xi_{3}^{2}+\xi_{4}\right) \frac{\partial}{\partial \xi_{2}}\\
			&+\left(\xi_{2}-\xi_{3}^{2}+\xi_{4}\right) \frac{\partial}{\partial \xi_{3}}.
	\end{aligned}
	\end{equation}
	In noting the form of the right-hand side of the above equation, we select 
	\begin{equation}\nonumber
		G^3(\xi)=\frac{\partial}{\partial \xi_{1}}+2 \xi_{3} \frac{\partial}{\partial \xi_{2}}+\frac{\partial}{\partial \xi_{3}}\;.
	\end{equation}
	Thanks to the choise for $G^3$, $\operatorname{ad}_{G^3} F$ is of such a simple form that we immediately take     
	\begin{equation}\nonumber
		G^2(\xi)=\operatorname{ad}_{G^3} F=\frac{\partial}{\partial \xi_{2}}.
	\end{equation}
	After finishing the computation of $\operatorname{ad}_{G^2} F$, as shown in \eqref{eq_adG2F}, 
	\begin{figure*}
		\begin{equation}\label{eq_adG2F}
			\begin{aligned}
				{\rm{ad}}_{F} G^{2}=&\left(2 \xi_{1}-2 \xi_{3}+\xi_{4}+3\left(\xi_{2}-\xi_{3}^{2}+\xi_{4}\right)^{2}\right) \frac{\partial}{\partial \xi_{1}} 
				+\left(\xi_{1}+2 \xi_{3}\left(\xi_{1}-\xi_{3}+\xi_{4}\right)-\xi_{3}+3\left(\xi_{2}-\xi_{3}^{2}+\xi_{4}\right)^{2}+1\right) \frac{\partial}{\partial \xi_{2}} \\
				&+\left(\xi_{1}-\xi_{3}+\xi_{4}\right) \frac{\partial}{\partial \xi_{3}}+\left(-\xi_{1}+\xi_{3}-3\left(\xi_{2}-\xi_{3}^{2}+\xi_{4}\right)^{2}\right) \frac{\partial}{\partial \xi_{4}} \\
				=&\left(\xi_{1}-\xi_{3}+3\left(\xi_{2}-\xi_{3}^{2}+\xi_{4}\right)^{2}\right) \frac{\partial}{\partial \xi_{1}}+G^{2}+\left(\xi_{1}-\xi_{3}+\xi_{4}\right) G^{3} 
				+\left(-\xi_{1}+\xi_{3}-3\left(\xi_{2}-\xi_{3}^{2}+\xi_{4}\right)^{2}\right) G^{4}
			\end{aligned}
		\end{equation}
		\hrulefill
	\end{figure*}
	we take		
	\begin{equation}\nonumber
		G^1(\xi)=\frac{\partial}{\partial \xi_{1}}.
	\end{equation}
	It is easy to verify that $D^i = {\rm{span}} \{G^i,\dots,G^4\} $, $i = 1, \dots, 4$,  are $5-i$ dimensional involutive distributions satisfying $ D^i = {\mathrm{span}} \{ \mathrm{ad}_{G^{i+1}} F, \hat D^{i+1}\} $ in an open set whose closure is a neighborhood of the origin. The Frobenius theorem \cite{lang2012differential} guarantees that there exists a change of coordinates $x=T(\xi)$ such that
	\begin{equation}\nonumber
		\begin{aligned}
			& G^j(T_i)  \ne 0, \;{\rm{for}} \;  i=j \\
			& G^j(T_i)   = 0, \;{\rm{for}} \;  i<j
		\end{aligned}
	\end{equation}
	where $i= 1,\dots,4$. Solving the above equations, we obtain a change of coordinates
	\begin{equation}\nonumber
		\begin{aligned}
			x_{1} & = \xi_{1}-\xi_{3} \\
			x_{2} & = \xi_{2}+\xi_{4}-\xi_{3}^{2} \\
			x_{3} & = \xi_{3} \\
			x_{4} & = \xi_{4}
		\end{aligned}
	\end{equation}
	and, in $x$-coordinates, \eqref{eq_nse1}  can be rewritten as 
	\begin{equation}\label{eq_nse1x}
		\begin{aligned}
			\dot{x}_{1}&=x_{2}^{3}+x_{1} x_{2} \\
			\dot{x}_{2}&=x_{3}+x_{2} \\
			\dot{x}_{3}&=x_{4} x_{2}+x_{2} x_{1}+x_{1} \\
			\dot{x}_{4}&=-x_{2}^{3}-x_{2} x_{1}+\left(1+(x_1+x_4)^{2}\right) u.
		\end{aligned}
	\end{equation}
	Examining the right-hand sides of the first three equations of \eqref{eq_nse1x}, this system is of type $[(0,3),(0,0,1),(0,1,0,1)]$, and is also of type $\left[\kern-0.15em\left[ {\{ (0,3),(1,1)\} ,\{ (0,0,1)\} ,\{(0,1,0,1)\} } \right]\kern-0.15em\right]$.  

\end{example}

\subsection{Conditions for a System to be Equivalent to a Given Type of Lower Triangular Form}
In this subsection, we investigate what condition is met to judge that a nonlinear system is equivalent to a specific type of lower triangular system without taking an equivalent transformation. Let us start with the following definition.\par

\begin{definition}\label{defi_leftless}
	$\alpha$ is a multi-index and $\beta$ is a proper $m_\beta$-multi-index with $1 \le m_{\beta} \le n$. $\alpha$ is said to be left equal to $\beta$, denoted by $\alpha  =_l \beta $, if $\alpha_i = \beta_i $ for all $i = 1,2, \dots ,m_\beta$; $\alpha$ is said to be left less than $\beta$, denoted by $\alpha  <_l \beta $, if $\alpha_i \le \beta_i $ holds for all $i = 1, \dots ,m_\beta$ and there exists at least one $j \in \{1, \dots ,m_\beta\}$ such that $\alpha_j < \beta_j$. We also define that $0$ is the only multi-index left equal to $0$ and there is no multi-index left less than $0$. Moreover, if $\alpha <_l \beta$ or $\alpha =_l \beta$, we write $\alpha \le_l \beta$.  
\end{definition}
\begin{example}
		According to the definition above, we have $(1,1,2) =_l (1,1)$ and $(1,1,1,1) <_l (1,1,2)$. 
\end{example}\par
From the above definition, it is trivial to verify the following lemma.
\begin{lemma}\label{lemma_partial}
	$p(x_1,\dots,x_n)$ is a smooth function, $\alpha=(\alpha_1,\dots,\alpha_m)$ is a proper $m$-multi-index with $m \le n$, and there is no multi-index of $p(x_1,\dots,x_n)$ left less than $\alpha$. Then, any multi-index of $\partial p / \partial x_k$ for $k \in \{m,\dots,n\}$ is not left less than $\alpha'=({\alpha _1}, \dots ,{\alpha _{m - 1}},{\alpha _m} - 1)$, and $\alpha'$ is a multi-index of $\partial p / \partial x_k (k \ge m)$ if and only if $\alpha$ is a multi-index of $p$ and $k=m$.
\end{lemma}\par
Using \eqref{eq_Leibniz},  we obtain the lemma as follows. 
\begin{lemma}\label{lemma_multiply}
	$p(x_1,\dots,x_n)$ and $q(x_1,\dots,x_n)$ are smooth functions with $p(0)=0$, and $\alpha$ is a proper $i$-multi-index. Suppose $L_\alpha=\{\beta|\beta<_l\alpha\}$ and $L_\alpha \bigcap {\cal I}(p) = \emptyset$. Then (i) $L_\alpha \bigcap {\cal I}(p \cdot q) = \emptyset$; (ii) $\alpha \in {\cal I}(p \cdot q)$ if and only if $q(0) \neq 0$ and $\alpha \in {\cal I}(p)$; (iii) For some $\alpha'=_l \alpha$, $\alpha' \in {\cal I}(p \cdot q)$ if and only if there exists a multi-index of $q$ left less than the proper $i$-multi-index $(0,\dots,0,1)$ and there exists $\bar\alpha =_l \alpha$ satisfying $\bar\alpha \in {\cal I}(p)$.\par	
\end{lemma}

Next, we present a differential geometric lemma that is useful for the further discussion in this subsection.
\begin{lemma}\label{lemm_Y}
	$Y(\xi)$ is a smooth vector field. There exists a change of coordinates $x=T(\xi)$ such that $Y$, in $x$-coordinates, can be expressed as
	\begin{equation}\label{eq_yk}
		\begin{aligned}
			Y(x)&=\sum\limits_{i = 1}^{k - 1} {Y_i(x_1, \dots ,x_{k - 1})\frac{\partial}{\partial x_i}} + \\
			&{Y_k}({x_1}, \dots ,{x_k})\frac{\partial }{{\partial {x_k}}}
			+ \sum\limits_{i = k + 1}^n {{Y_i}({x_1}, \dots ,{x_n})\frac{\partial }{{\partial {x_i}}}} 
		\end{aligned}
	\end{equation}
	if and only if there exist smooth vector fields $X^1(\xi),\dots,X^n(\xi)$ such that $D^l = {\rm{span}}  \{X^l, \dots ,X^n\}$, $l=n,\dots,1$, are $n-l+1$ dimensional involutive distributions and
	\begin{equation}\nonumber
		[{X^l},Y] \in \left\{ {\begin{matrix}
				{{D^{k + 1}}}&{k+1 \le l \le n}\\
				{{D^k}}&{l = k}
		\end{matrix}} \right..
	\end{equation}
\end{lemma}
\begin{proof}
	The necessity is clear, we only prove the sufficiency here. According to the Frobenius theorem, we can find a change of coordinates $x=T(\xi)$ such that $D^l=\mathrm{span}\{\partial/\partial x_l,\dots,\partial/\partial x_n\}$ and $X^l =  {\textstyle\sum_{i = l}^{n}{X^l_i(x) \partial / \partial x_i}} $ for $l = n,\dots,1$. Let $Y(x) = {\textstyle \sum_{i = 1}^{n}{Y_i(x) \partial / \partial x_i}}$. Noting that $[{X^l},Y]={\textstyle\sum_{i = l}^{n} \left( {X^l_i[\partial / \partial x_i,Y]}-Y(X^l_i)\partial / \partial x_i \right) } \in D^{k+1}$ for $l=n,\dots,k+1$, we have $\partial Y_j / \partial x_l =0$ for all $j = 1,\dots ,k$ and $l = n,\dots ,k+1$. Additionally, $[{X^k},Y] \in D^k$ implies that $\partial Y_j / \partial x_k =0$ for any $j = 1,\dots ,k-1$. Thus \eqref{eq_yk} holds.  	
\end{proof}\par

Let $X$ be a vector field, $Y^1, \dots, Y^m$ a family of vector fields, and $\alpha = (\alpha_1,\dots,\alpha_k)$ a $k$-multi-index. We denote, for $i=1,\dots,m$ and an integer $j \ge 0$, $\mathrm{ad}_{Y^i}^0 X = X$, $\mathrm{ad}_{Y^i}^1 X = \mathrm{ad}_{Y^i} X$, $\mathrm{ad}_{Y^i}^{j+1} X = \mathrm{ad}_{Y^i} \mathrm{ad}_{Y^i}^j X$, and $\mathrm{ad}_Y^\alpha X=[{Y^\alpha},X] = \mathrm{ad}_{Y^1}^{\alpha_1} \dots \mathrm{ad}_{Y^k}^{\alpha_k}X$. Now we are ready to state several properties of lower triangular forms.\par
\begin{proposition}\label{prop_essindofsys}
	System \eqref{eq_tri} is of type $[\kern-0.15em[ {\cal E}_2(f_1),\dots,{\cal E}_n(f_{n-1}) ]\kern-0.15em]$. Let
	\begin{equation}\nonumber
		{D}^i={\rm{span}} \left\{ \frac{\partial}{\partial x_i},\dots,\frac{\partial}{\partial x_n} \right\}, i=n,\dots,1,
	\end{equation}
	and
	\begin{equation}\label{eq_triY}
	\begin{aligned}
		&{Y^{i + 1}} = \sum\limits_{k = i + 1}^n {Y_k^{i + 1}(x)\frac{\partial }{\partial {x_k}}} \\
		&{Y^j} = \sum\limits_{k = j}^{i - 1} {Y_k^j({x_1}, \dots ,{x_{i - 1}}) \frac{\partial }{\partial {x_k}}}  + Y_i^j({x_1}, \dots ,{x_i})\frac{\partial }{\partial {x_i}} \\
		&\qquad + \sum\limits_{k = i + 1}^n {Y_k^j(x)\frac{\partial }{\partial {x_k}}},\; j=1,\dots,i,		
	\end{aligned}
	\end{equation}
	where ${Y_{i + 1}^{i + 1}(0) \ne 0}$ and ${Y_j^j(0) \ne 0}$. Then $\epsilon \in {\cal E}_{i+1}(f_i)$ if and only if 
	\begin{equation}\label{eq_epsilon}
		\mathrm{ad}_Y^\epsilon f(0) \notin {D}^{i + 1}(0)
	\end{equation}
	and
	\begin{equation}\label{eq_alpha}
		\mathrm{ad}_Y^\alpha f(0) \in {{D}^{i + 1}}(0)
	\end{equation}
	for every proper $(i+1)$ multi-index $\alpha \prec \epsilon$. In addition, a proper $(i+1)$-multi-index $\zeta$ and all the $(i+1)$-multi-indices that can generate $\zeta$ do not belong to ${\cal I}_{i+1}(f_i)$ if and only if 
	\begin{equation}\label{eq_alpha2}
		\mathrm{ad}_Y^\alpha f(0) \in {{D}^{i + 1}}(0)
	\end{equation}
	for every proper $(i+1)$ multi-index $\alpha \preceq \zeta$.
\end{proposition}
\begin{proof}
	We first calculate ${\rm{ad}}_Y^{\theta} f$, where $\theta=(\theta_1,\dots,\theta_{i+1})$ is a proper $(i+1)$-multi-index, step by step. Let ${X^{i + 1,\theta _{i + 1}}} = f$. Compute the following Lie brackets
	\begin{equation}\nonumber
	\begin{aligned}
		&{X^{i + 1,\theta _{i + 1}-1}} = [{Y^{i + 1}},{X^{i + 1,\theta _{i + 1}}}] =\sum_{j=i}^{n} X_{j}^{i + 1,\theta _{i + 1}-1}(x)\frac{\partial }{{\partial {x_j}}} \\
		&\quad \vdots\\
		&{X^{i + 1,0}} = [{Y^{i + 1}},{X^{i+1,1}}] = \sum_{j=i}^{n} X_{j}^{i + 1,0}(x)\frac{\partial }{{\partial {x_j}}}
	\end{aligned}
	\end{equation}
	where $X_{j}^{i + 1,k}(x)$, $j=i,\dots,n$ and $k=\theta _{i + 1}-1,\dots,0$, are all smooth functions, especially
	\begin{equation}\label{eq_adi+1}
	\begin{aligned}
		&X_{i}^{i + 1,\theta _{i + 1}-1}(x) = \frac{{\partial X_i^{i + 1,\theta _{i + 1}}}}{{\partial {x_{i+1}}}}Y_{i+1}^{i + 1}=\frac{{\partial f_i}}{{\partial {x_{i+1}}}}Y_{i+1}^{i + 1}\\
		&X_{i}^{i + 1,\theta _{i + 1}-2}(x) = \sum_{j=i+1}^{n} \frac{{\partial X_{i}^{i + 1,\theta _{i + 1}-1}}}{{\partial {x_j}}}Y_{j}^{i+1} \\
		&\quad \vdots \\
		&X_{i}^{i + 1,0}(x)= \sum_{j=i+1}^{n} \frac{{\partial X_{i}^{i + 1,1}}}{{\partial {x_j}}}Y_{j}^{i + 1} .
	\end{aligned}
	\end{equation}
	Let ${X^{l,\theta_l}} = {X^{l + 1,0}}$ for $l = i, \dots ,1$. Proceeding in the same manner, one can calculate
	\begin{equation}\nonumber
	\begin{aligned}
		&{X^{l,\theta_l-1}} = [{Y^l},{X^{l,\theta_l}}]= \sum_{j=i}^{n} X_{j}^{l,\theta_l-1}(x)\frac{\partial }{{\partial {x_j}}}\\
		&\quad \vdots\\
		&{X^{l,0}} = [{Y^l},{X^{l,1}}]=\sum_{j=i}^{n} X_{j}^{l,0}(x)\frac{\partial }{{\partial {x_j}}}
	\end{aligned}
	\end{equation}
	where $X_{j}^{l,k}(x)$, $l=i,\dots,1$, $j=i,\dots,n$, and $k=\theta_l-1,\dots,0$, are all smooth functions with
	\begin{equation}\label{eq_adi_1}
	\begin{aligned}
		&X_{i}^{l,\theta_l-1}(x) =   - \frac{{\partial Y_i^l}}{{\partial {x_i}}}{X_{i}^{l,\theta_l}} + \sum_{j=l}^{n}{\frac{{\partial {X_{i}^{l,\theta_l}}}}{{\partial {x_{j}}}}Y_{j}^l}  \\
		&\quad \vdots \\
		&X_{i}^{l,0}(x) = - \frac{{\partial Y_i^l}}{{\partial {x_i}}}{X_{i}^{l, 1}} + \sum_{j=l}^{n} {\frac{{\partial {X_{i}^{l,1}}}}{{\partial {x_{j}}}}Y_{j}^l}\;.
	\end{aligned}
	\end{equation}	\par
	Assuming $\theta=(\theta_1,\dots,\theta_{i+1})$ is a multi-index belonging to ${\cal E}_{i+1}(f_i)$, we now prove ${\rm{ad}}_Y^{\theta} f(0) \notin {D}^{i + 1}(0)$. For the sake of convenience, we denote $\theta^{l,k} = (\theta_1,\dots,\theta_{l-1},k)$ for $l=i+1,\dots,1$ and $k=\theta_l,\dots,0$. So we have $\theta^{l,0} =\theta^{l-1,\theta_{l-1}}$ for $l=i+1,\dots,2$. It is clear that $\theta^{i+1,\theta_{i + 1}-1}$ is a multi-index of ${\partial X_i^{i + 1,\theta_{i+1}}} / {\partial {x_{i+1}}}$. Let $\beta$ be an $(i+1)$-multi-index with $\beta <_l \theta^{i+1,\theta_{i + 1}-1}$, then we can assert that $\beta \notin {\cal I}({\partial X_i^{i + 1,\theta_{i+1}}} / {\partial {x_{i+1}}})$ because otherwise one would exhibit $(\beta_1,\dots,\beta_i,\beta_{i+1}+1) \in {\cal I}_{i+1}(f_i)$ and $(\beta_1,\dots,\beta_i,\beta_{i+1}+1) \prec \theta$, which is contradictory with $\theta \in {\cal E}_{i+1}(f_i)$. From the first equation of \eqref{eq_adi+1} and lemma \ref{lemma_multiply}, $\theta^{i+1,\theta_{i + 1}-1} \in {\cal I}(X_{i}^{i + 1,\theta_{i+1}-1}(x))$ holds and there is no element of ${\cal I}(X_{i}^{i + 1,\theta_{i+1}-1}(x))$ which is left less than $\theta^{i+1,\theta_{i + 1}-1}$. Suppose, for any $k \in \{\theta_{i+1}-1,\dots,1\}$, $X_{i}^{i + 1,k}(x))$ satisfies $\theta^{i+1,k} \in {\cal I}(X_{i}^{i + 1,k}(x))$ and $\beta \not\in {\cal I}(X_{i}^{i + 1,k}(x))$ for all $\beta <_l \theta^{i+1,k}$. It follows from \eqref{eq_adi+1}, lemma \ref{lemma_partial}, and lemma \ref{lemma_multiply} that $\theta^{i+1,k-1} \in {\cal I}(X_{i}^{i + 1,k-1}(x))$ and $\beta \not\in {\cal I}(X_{i}^{i + 1,k-1}(x))$ for all $\beta <_l \theta^{i+1,k-1}$. Consider $X_{i}^{l,k}(x)$ for $l \in \{i,\dots,1\}$ and $k \in \{ \theta_l-1,\dots, 0\}$ given by \eqref{eq_adi_1}. Assume that $\theta^{l,k+1} \in {\cal I}(X_{i}^{l,k+1}(x))$ and $\beta \not\in {\cal I}(X_{i}^{l,k+1}(x))$ for all $\beta <_l \theta^{l,k+1}$. Take account of $\theta^{l,k} <_l \theta^{l,k+1}$ and lemma \ref{lemma_multiply}, $\beta \not\le_l \theta^{l,k}$ holds for any $\beta \in {\cal I} ( {\partial Y_i^l}/{{\partial {x_i}}}\cdot{X_{i}^{l,k+1}} )$. Using lemma \ref{lemma_partial} and \ref{lemma_multiply}, we have $\theta^{l,k} \in {\cal I} ( {{\partial X_i^{l,k+1}}}/{{\partial {x_l}}}\cdot{Y_{l}^{l}} )$, $\theta^{l,k} \not\in {\cal I} ( {{\partial X_i^{l,k+1}}}/{{\partial {x_j}}}\cdot{Y_{j}^{l}} )$ for $j=l+1,\dots,n$, and $\beta \not\in {\cal I} ( {{\partial X_i^{l,k+1}}}/{{\partial {x_{j'}}}}\cdot{Y_{j'}^{l}} )$ for any $\beta <_l \theta^{l,k}$ and $j'=l,\dots,n$. This means $\theta^{l,k}  \in {\cal I}(X_{i}^{l,k}(x))$ and $\beta \not\in {\cal I}(X_{i}^{l,k}(x))$ for all $\beta <_l \theta^{l,k}$. Especially,  $\theta^{1,0}=(0)  \in {\cal I}(X_{i}^{1,0}(x))$, which implies $X_{i}^{1,0}(0) \ne 0$ and ${\rm{ad}}_Y^{\theta} f(0) \notin {D}^{i + 1}(0)$. \par
	In a similar way, we can prove that if $\theta$ is a proper $(i+1)$ multi-index such that $\beta \not\in {\cal I}^{i+1}(f_i)$ holds for every $\beta \preceq \theta$ then $\gamma \not\in {\cal I}(X_{i}^{l,k}(x))$ for any $\gamma \le_l \theta^{l,k} = (\theta_1,\dots,\theta_{l-1},k)$ $(l=i+1,\dots,1 $ and $k = \theta_l-1,\dots,0)$. Hence ${\rm{ad}}_Y^{\theta} f(0) \in {D}^{i + 1}(0)$.\par
	Consider $\mathrm{ad}_Y^\epsilon f$  and $\mathrm{ad}_Y^\alpha f$ with $\epsilon \in {\cal E}_{i+1}(f_i)$ and $\alpha \prec \epsilon$. Directly from the previous discussions, one can obtain \eqref{eq_epsilon} and \eqref{eq_alpha}. We now prove that \eqref{eq_epsilon} and \eqref{eq_alpha} imply $\epsilon \in {\cal E}_{i+1}(f_i)$. Suppose the proper $(i+1)$-multi-index $\alpha$ introduced in \eqref{eq_alpha} satisfies $\alpha \in {\cal I}_{i+1}(f_i)$. There exists a proper $(i+1)$ multi-index $\alpha' \preceq \alpha$ and $\alpha' \in {\cal E}_{i+1}(f_i)$. Since it has been proved that \eqref{eq_epsilon} holds when $\epsilon \in {\cal E}_{i+1}(f_i)$, we obtain $\mathrm{ad}_Y^{\alpha'} f(0) \notin {D}^{i + 1}(0)$, which contradicts \eqref{eq_alpha}. Therefore, $\alpha \not\in {\cal I}_{i+1}(f_i)$ must be true. We next consider the $(i+1)$ multi-index $\epsilon$. It is clear that $\epsilon \in {\cal I}_{i+1}(f_i)$ implies $\epsilon \in {\cal E}_{i+1}(f_i)$. If $\epsilon \not\in {\cal E}_{i+1}(f_i)$ were true then, with $\alpha \not\in {\cal I}_{i+1}(f_i)$ for every  $\alpha \prec \epsilon$ in mind, $\mathrm{ad}_Y^{\epsilon} f(0) \in {D}^{i + 1}(0)$ would hold, which also contradicts \eqref{eq_alpha}. Thus, we conclude that $\epsilon \in {\cal E}_{i+1}(f_i)$.\par
	It has been proved that \eqref{eq_alpha2} holds when $\alpha \notin {\cal I}_{i+1}(f_i)$ for all $\alpha \preceq \zeta$. On the other hand, the existence of some $\alpha \preceq \zeta$ satisfying  $\alpha \in {\cal I}_{i+1}(f_i)$, together with \eqref{eq_alpha2},  obviously contradicts \eqref{eq_epsilon}. Hence, \eqref{eq_alpha2} implies $\alpha \notin {\cal I}_{i+1}(f_i)$ for all $\alpha \preceq \zeta$.
\end{proof}\par
\begin{example}
	This counterexample shows that the above proposition is not valid if one modifies \eqref{eq_triY} to ${Y^j} = \sum_{k = j}^{n} Y_k^j(x)  \partial / \partial x_k$. Consider the following system
	\begin{equation}\nonumber
		\begin{aligned}
			&{{\dot x}_1} = x_1^3{x_2}\\
			&{{\dot x}_2} = {x_3}\\
			&{{\dot x}_3} = {x_2} + v\;.
		\end{aligned}
	\end{equation}
	Let $Y_3 = {\partial }/{\partial x_3}$, $Y_2 = {\partial }/{\partial x_2}$, and $Y_1 = (1 + x_3){\partial }/{\partial x_1} + (x_2 - x_1){\partial }/{\partial x_3}$. Here $Y_1$ does not satisfy \eqref{eq_triY}.
	One can obtain
	\begin{equation}\nonumber
		\mathrm{ad}_{Y_1} \mathrm{ad}_{Y_2}f = \left( {3x_1^2\left( {{x_3} + 1} \right) - 1} \right)\frac{\partial }{{\partial {x_1}}} + x_1^3\frac{\partial }{{\partial {x_3}}}
	\end{equation}
	and ${\rm{a}}{{\rm{d}}_{{Y_1}}}{\rm{a}}{{\rm{d}}_{{Y_2}}}f(0) \not\in \mathrm{span}\{\partial/{\partial x_2}, \partial/{\partial x_3}\}$. But ${\cal L}_2(x_1^3{x_2}) $ $= (3,1) $.
\end{example}\par

	Combining Proposition \ref{prop_essindofsys}, Lemma \ref{lemm_Y}, and \eqref{eq_pushforward}, it is easy to see the next two theorems.
\begin{theorem}\label{theo_sysleast}
	System \eqref{eq_s} is locally equivalent to a lower triangular form satisfying ${\cal L}_{i+1}(f_i)=\alpha^i$ $(i=1,\dots,n-1)$, namely \eqref{eq_lea_sys}, via a feedback \eqref{eq_fb} and a change of coordinates \eqref{eq_ct} if and only if the following conditions are satisfied:\par
	(i) System \eqref{eq_s} is locally feedback equivalent to \eqref{eq_tri}. \par
	(ii) Suppose $D^l$, $l=1,\dots,n+1$, are $n-l+1$ dimensional involutive distributions defined in Theorem \ref{theorem_tri} and $X^l$, $l=1,\dots,n$, are smooth vector fields satisfying $X^l \in D^l$ and $X^l(0) \not\in D^{l+1}(0)$. Let $Y^{i}= (Y^{i,1},\dots,Y^{i,i+1})$, $i=1,\dots,n-1$, be tuples of smooth vector fields satisfying $Y^{i,j} \in D^j$, $Y^{i,j}(0) \not\in D^{j+1}(0)$, and	
	\begin{equation}\nonumber
		[{X^l},Y^{i,j}] \in \left\{ {\begin{matrix}
				D^{i+1}& i+1\le l \le n\\
				D^{i}&l = i
		\end{matrix}} \right.
	\end{equation}
	for $j=1,\dots,i+1$. Then ${\mathrm{ad}_{Y^i}^{\alpha^i} F}(0) \notin {D^{i + 1}}(0)$ and ${\mathrm{ad}_{Y^i}^\alpha F}(0) \in {D^{i + 1}}(0)$ for every proper $(i+1)$-multi-index $\alpha  \lessdot \alpha^i$.
\end{theorem}
\begin{remark}
	$X^1,\dots,X^n$ given in the above theorem obviously satisfy $\mathrm{span}\{X^l,\dots,X^n\} = D^l$ in a neighborhood of the origin. It is not difficult to find $X^1,\dots,X^n$ when $D^l$, $l=1,\dots,n$, are known. 
\end{remark}
\begin{remark}
	The necessary and sufficient condition introduced in the above theorem for a nonlinear system to be equivalent to a $p$-normal form is consistent with the condition given in \cite{respondek2003transforming} if we take $Y^{n-1,n}=G$ and $Y^{i,i+1}=\mathrm{ad}_{Y^{i+1,i+2}}^{p_{i+1}}F$ for $i=n-2,\dots,1$.
\end{remark}

\begin{theorem}\label{theorem_type}
	System \eqref{eq_s} is locally equivalent to a lower triangular form taking the form \eqref{eq_ess_sys} via a feedback \eqref{eq_fb} and a change of coordinates \eqref{eq_ct} if and only if the following conditions are satisfied:\par
	(i) System \eqref{eq_s} is locally feedback equivalent to \eqref{eq_tri}. \par
	(ii) Suppose $D^l$, $l=1,\dots,n+1$, are $n-l+1$ dimensional involutive distributions defined in Theorem \ref{theorem_tri} and $X^l$, $l=1,\dots,n$, are smooth vector fields satisfying $X^l \in D^l$ and $X^l(0) \not\in D^{l+1}(0)$. Let $Y^{i}= (Y^{i,1},\dots,Y^{i,i+1})$, $i=1,\dots,n-1$, be tuples of smooth vector fields satisfying $Y^{i,j} \in D^j$, $Y^{i,j}(0) \not\in D^{j+1}(0)$, and	
	\begin{equation}\nonumber
		[{X^l},Y^{i,j}] \in \left\{ {\begin{matrix}
		D^{i+1}& i+1\le l \le n\\
		D^{i}&l = i
		\end{matrix}} \right. 
	\end{equation}
	for $j=1,\dots,i+1$. Then for every $\epsilon \in {\cal E}_{i+1}(f_i)$ and every $\zeta  \in {\cal A}_{i+1} \setminus {\cal G}_{i+1}({\cal E}_{i+1}(f_i))$, where ${\cal A}_{i+1}$ is the set consisting of all the proper $i$-multi-indices, the relations $\mathrm{ad}_{Y^i}^\epsilon F(0) \notin D^{i + 1}(0)$ and $\mathrm{ad}_{Y^i}^{\zeta} F(0) \in D^{i + 1}(0)$ hold.
\end{theorem}\par

\begin{remark}
	 When the proper $(i+1)$-multi-index $(0,\dots,0,k) \in {\cal E}_{i+1}(f_i)$, we know that ${\cal A}_{i+1} \setminus {\cal G}_{i+1}({\cal E}_{i+1}(f_i))$ is finite from Proposition \ref{prop_finiteset}. To check the condition (ii) in the above theorem, we only need to calculate Lie brackets a finite number of times. But when the proper $(i+1)$-multi-index $(0,\dots,0,k) \not\in {\cal E}_{i+1}(f_i)$ for all the positive integer $k$, ${\cal A}_{i+1} \setminus {\cal G}_{i+1}({\cal E}_{i+1}(f_i))$ is infinite. Although it follows that we may need to compute Lie brackets infinitely many times in the case, this is acceptable because one may also have to check infinite many $i$-multi-indices of $f_i$ to find ${\cal E}_{i+1}(f_i)$.
\end{remark}

\begin{remark} \label{rem_Yi}
	We now consider how to obtain $Y^i$ required in Theorem \ref{theo_sysleast} and \ref{theorem_type}. $Y^{i,i+1}$ can be selected as any smooth vector field belonging to $D^{i+1}$ such that $Y^{i+1}(0) \not\in D^{i+2}(0)$. Let, for $j =1,\dots\,i$, $Y^{i,j} = \textstyle \sum_{k=j}^{n} h^{i,j}_k(\xi) X^k$ where $h^{i,j}_k(\xi)$ are undetermined smooth functions. Note that
	\begin{equation}\label{eq_XtoX}
		[X^{l_1},X^{l_2}] = a^{l_1,l_2}_{l'}(\xi)X^{l_2}+\dots+a^{l_1,l_2}_{n}(\xi)X^{n} 		
	\end{equation}
	where $l_1$ and $l_2$ are integers belonging to $\{1,\dots,n\}$, ${l'} = \mathrm{min}(l_1,l_2)$, and $a^{l_1,l_2}_{l'}(\xi),\dots, a^{l_1,l_2}_{n}(\xi)$ are smooth functions satisfying $a^{l_1,l_2}_k(\xi)=-a^{l_2,l_1}_k(\xi)$ and $a^{l,l}_k(\xi)=0$. Let us calculate the following Lie bracket, for $l=i+1,\dots,n$,
	\begin{equation}\nonumber
	\begin{aligned}
		&[X^l,Y^{i,j}]= \sum_{k=j}^{i}[X^l,h^{i,j}_k X^k]+\sum_{k={i+1}}^{n}[X^l,h^{i,j}_k X^k]\\
		&\;=\sum_{k=j}^{i} \left(X^l(h^{i,j}_k)X^k + h^{i,j}_k[X^l,X^k]\right)+\sum_{k={i+1}}^{n}[X^l,h^{i,j}_k X^k]\\
		&\;= \sum_{k=j}^{i} \left(X^l(h^{i,j}_k)X^k + h^{i,j}_k \sum_{\hat k=k}^{i} a^{l,k}_{\hat k}X^{\hat k}\right) +\\
		&\quad\;\, \sum_{k={j}}^{i} \left(h^{i,j}_k \sum_{\hat k={i+1}}^{n} a^{l,k}_{\hat k}X^{\hat k}\right)+\sum_{k={i+1}}^{n}[X^l,h^{i,j}_k X^k]\\		
		&\;= \sum_{k=j}^{i} \left(X^l(h^{i,j}_k) + \sum_{k'=j}^{k} h^{i,j}_{k'} a^{l,k'}_{k} \right)X^k +\\
		&\quad\;\, \sum_{k={j}}^{i} \left(h^{i,j}_k \sum_{\hat k={i+1}}^{n} a^{l,k}_{\hat k}X^{\hat k}\right)+\sum_{k={i+1}}^{n}[X^l,h^{i,j}_k X^k] \in D^{i+1}.	
	\end{aligned}
	\end{equation}
	Similarly, we have
	\begin{equation}\nonumber
	\begin{aligned}
		&[X^i,Y^{i,j}]= \sum_{k=j}^{i-1} \left(X^l(h^{i,j}_k) + \sum_{k'=j}^{k} h^{i,j}_{k'} a^{l,k'}_{k} \right)X^k +\\
		&\quad\;\, \sum_{k={j}}^{i-1} \left(h^{i,j}_k \sum_{\hat k={i}}^{n} a^{l,k}_{\hat k}X^{\hat k}\right)+\sum_{k={i}}^{n}[X^l,h^{i,j}_k X^k] \in D^{i}.
	\end{aligned}	
	\end{equation}
	Thus, $h^{i,j}_k(\xi)$, $k=j,\dots,i$, can be determined from the equations
	\begin{equation}\nonumber
		 X^l(h^{i,j}_k) + \sum_{k'=j}^{k} h^{i,j}_{k'} a^{l,k'}_{k} =0
	\end{equation}
	where $k=j,\dots,i$ when $l=i+1,\dots,n$, and $k=j,\dots,i-1$ when $l=i$. The existence of the solution of these equations is guaranteed by Proposition \ref{prop_essindofsys} and Lemma \ref{lemm_Y}. Additionally, $h^{i,j}_k(\xi)$, $k=i+1,\dots,n$, can be chosen to be arbitrary smooth functions.
\end{remark}\par
 To determine whether a nonlinear system can be transformed into a specific type of lower triangular form by using the previous two theorems, appropriate vector fields $Y^{i,j}$, $i=1,\dots,n-1$ and $j=1,\dots,i+1$, are required. Partly because there are so many vector fields to find out, this is not an easy task. The next two corollaries are the reduced versions of Theorem \ref{theo_sysleast} and \ref{theorem_type}, respectively. The following lemma can be proved in a way similar to the proof of Lemma \ref{lemm_Y}.
\begin{lemma}\label{lemma_Yn}
	$\{X^1(\xi), \dots ,X^n(\xi)\}$ and $\{Y^1(\xi), \dots ,Y^n(\xi)\}$ are two sets of  nonsingular vector fields such that $D^k = {\rm{span}}  \{{X^k}, \dots ,{X^n}\} ={\rm{span}}  \{{Y^k}, \dots ,{Y^n}\}$, $k=n,\dots,1$, are $n-k+1$ dimensional involutive distributions. Then there exists a change of coordinates $x=T(\xi)$ such that in $x$-coordinates
	\begin{equation}\nonumber
		{Y^k(x)} = \sum\limits_{i = k}^{n}{Y_i^k({x_1}, \dots ,{x_i})\frac{\partial }{{\partial {x_i}}}}
	\end{equation}
	if and only if the relation
	\begin{equation}\nonumber
		[{X^i},{Y^j}] \in {D^i}
	\end{equation}
	holds for any pair of $i,j = 1,\dots,n$ satisfying $i>j$.
\end{lemma}

The next two corollaries follow at once from the previous two theorems and the above lemma.
\begin{corollary}\label{col_least}
	System \eqref{eq_s} is locally equivalent to a lower triangular form satisfying ${\cal L}_{i+1}(f_i)=\alpha^i$ $(i=1,\dots,n-1)$, namely \eqref{eq_lea_sys}, via a feedback \eqref{eq_fb} and a change of coordinates \eqref{eq_ct} if and only if the following conditions are satisfied:\par
	(i) System \eqref{eq_s} is locally feedback equivalent to \eqref{eq_tri}. \par
	(ii)  Suppose $D^l$, $l=1,\dots,n+1$, are $n-l+1$ dimensional involutive distributions defined in Theorem \ref{theorem_tri} and $X^l$, $l=1,\dots,n$, are smooth vector fields satisfying $X^l \in D^l$ and $X^l(0) \not\in D^{l+1}(0)$. Let  $Y= (Y^1,\dots,Y^n)$ be a tuple of smooth vector fields such that $D^k = {\mathrm{span}} \{Y^k,\dots,Y^n\}$ for $k=n,\dots,1$ and $[{X^l},{Y^k}] \in {D^l}$ for all $l > k$. Then ${\mathrm{ad}_{Y}^{\alpha^i} F}(0) \notin {D^{i + 1}}(0)$ and ${\mathrm{ad}_{Y}^\alpha F}(0) \in {D^{i + 1}}(0)$ for every proper $(i+1)$-multi-index $\alpha  \lessdot \alpha^i$.
\end{corollary}

\begin{corollary}\label{col_type}
	System \eqref{eq_s} is locally equivalent to \eqref{eq_ess_sys} via a feedback \eqref{eq_fb} and a change of coordinates \eqref{eq_ct} if and only if the following conditions are satisfied:\par
	(i) System \eqref{eq_s} is locally feedback equivalent to \eqref{eq_tri}. \par
	(ii)  Suppose $D^l$, $l=1,\dots,n+1$, are $n-l+1$ dimensional involutive distributions defined in Theorem \ref{theorem_tri} and $X^l$, $l=1,\dots,n$, are smooth vector fields satisfying $X^l \in D^l$ and $X^l(0) \not\in D^{l+1}(0)$. Let  $Y= (Y^1,\dots,Y^n)$ be a tuple of smooth vector fields such that $D^k = {\mathrm{span}} \{Y^k,\dots,Y^n\}$ for $k=n,\dots,1$ and $[{X^l},{Y^k}] \in {D^l}$ for all $l > k$. Then for every $\epsilon \in {\cal E}_{i+1}(f_i)$ and every $\zeta  \in {\cal A}_{i+1} \setminus  {\cal G}_{i+1}({\cal E}_{i+1}(p))$, where ${\cal A}_{i+1}$ is the set consisting of all the proper $i$-multi-indices, the relations $\mathrm{ad}_Y^\epsilon F(0) \notin D^{i + 1}(0)$ and $\mathrm{ad}_Y^{\zeta} F(0) \in D^{i + 1}(0)$ hold.
\end{corollary}

\begin{remark}\label{rem_y}
	$Y$ mentioned in the previous two corollaries can be found by a method similar to Remark \ref{rem_Yi}. $Y^{n}$ can be selected as a smooth vector field belonging to $D^{n}$ with $Y^{n}(0) \ne 0$. Let, for $j =1,\dots\,n-1$, $Y^j = \textstyle \sum_{k=j}^{n} h_k^j(\xi) X^k$. Since $[X^l,Y^j] \in D^l$ for $l=j+1,\dots,n$, the functions $h_k^j(\xi)$, $k=j,\dots,n-1$, can be obtained from the solution of the equations
	\begin{equation}\nonumber
		X^l(h_k^j) + \sum_{k'=j}^{k} h_{k'}^j a^{l,k'}_{k} =0
	\end{equation}
	where every function $a^{l,k}_{k'}$ is defined by \eqref{eq_XtoX}. $h_n^j(\xi)$ can be any smooth function.
\end{remark}\par

\begin{example}
	Consider the system given by \eqref{eq_nse1}. By using the above corollary, we now show how to determine what type the system is without transforming it into a lower triangular form. Since it has been verified in Example \ref{exa_transformation} that this system satisfies the condition (i) introduced in Corollary \ref{col_type}, it is necessary to find four nonsingular vector fields $Y^4,\dots,Y^1$ such that, for $l=1,\dots,4$ and $k=1,\dots,l-1$, $[{X^l},{Y^k}] \in {D^l}$ where $D^l$ and $X^l=G^l$ are given in Example \ref{exa_transformation}. By using the method proposed in Remark \ref{rem_y}, let us take  
	\begin{equation}\nonumber
		\begin{aligned}
			&Y^{4}=-\frac{\partial}{\partial \xi_2}+\frac{\partial}{\partial \xi_4},\\			
			&Y^{3}=\frac{\partial}{\partial \xi_{1}}+\left(2 \xi_{3}-\xi_{4}\right) \frac{\partial}{\partial \xi_{2}}+\frac{\partial}{\partial \xi_{3}}+\xi_{4} \frac{\partial}{\partial \xi_{4}},\\			
			&Y^{2}=\xi_{3} \frac{\partial}{\partial \xi_{1}}+\left(2 \xi_{3}^{2}-\xi_{4}+1\right) \frac{\partial}{\partial \xi_{2}}+\xi_{3} \frac{\partial}{\partial \xi_{3}}+\xi_{4} \frac{\partial}{\partial \xi_{4}},
		\end{aligned}
	\end{equation}
	and	
	\begin{equation}\nonumber
		Y^{1}=\frac{\partial}{\partial \xi_{1}}-\xi_{4} \frac{\partial}{\partial \xi_{2}}+\xi_{4} \frac{\partial}{\partial \xi_{4}}.
	\end{equation}
	After computing several Lie brackets, it is straightforward to see that $[{X^l},{Y^k}] \in {D^l}$ for $l=1,\dots,4$ and $k=1,\dots,l-1$; that is, the condition (ii) introduced in Corollary \ref{col_type} is also satisfied. To simplify the notation, let \eqref{eq_nse1} be of a type $[\kern-0.15em[ E^1, E^2 ,E^3 ]\kern-0.15em]$, where $E^i$, $i=1,2,3$, are sets of proper $(i+1)$-multi-indices. To determine $E^3$, we first compute the following Lie bracket  
	\begin{equation}\nonumber
	\begin{aligned}
		\mathrm{ad}_{Y^{4}} F=&\left(\xi_{2}-\xi_{3}^{2}+\xi_{4}\right) \frac{\partial}{\partial \xi_{1}}+2 \xi_{3}\left(\xi_{2}-\xi_{3}^{2}+\xi_{4}\right) \frac{\partial}{\partial \xi_{2}}\\
		&+\left(\xi_{2}-\xi_{3}^{2}+\xi_{4}\right) \frac{\partial}{\partial \xi_{3}}.
	\end{aligned}
	\end{equation}
	Since $\mathrm{ad}_{Y^{4}} F(0) =0$, we have $(0,0,0,1) \not\in E^3$. After further computations, we obtain
	\begin{equation}\nonumber
	\begin{aligned}
		\mathrm{ad}&_{Y^{2}} \mathrm{ad}_{Y^{4}} F= \left(-\xi_{2}+\xi_{3}^{2}-\xi_{4}+1\right) \frac{\partial}{\partial \xi_{1}}+\\
		&2 \xi_{3}\left(-\xi_{2}+\xi_{3}^{2}-\xi_{4}+1\right) \frac{\partial}{\partial \xi_{2}}+\left(-\xi_{2}+\xi_{3}^{2}-\xi_{4}+1\right) \frac{\partial}{\partial \xi_{3}}
	\end{aligned}
	\end{equation}
	and 
	\begin{equation}\nonumber
		\mathrm{ad}_{Y^{1}} \mathrm{ad}_{Y^{4}} F=\mathrm{ad}_{Y^{3}} \mathrm{ad}_{Y^{4}} F=\mathrm{ad}_{Y^{4}} \mathrm{ad}_{Y^{4}} F=0.
	\end{equation}
	Seeing that $\mathrm{ad}_{Y^{2}} \mathrm{ad}_{Y^{4}} F(0) \not\in D^4$, it is clear that $\{(0,1,0,1)\} \in E^3$. Let $\alpha$ be a proper $4$-multi-index such that $\left| \alpha \right|>2$. Noting that $\mathrm{ad}_{Y}^\alpha F \ne 0$ implies $(0,1,0,1) \prec \alpha$, $E^3 = \{(0,1,0,1)\}$ holds. According to
	\begin{equation}\nonumber
	\begin{aligned}
		\mathrm{ad}&_{Y^{3}} F=\frac{\partial}{\partial \xi_{2}}+\xi_{4}\left(\xi_{2}-\xi_{3}^{2}+\xi_{4}\right)\left(\frac{\partial}{\partial \xi_{1}}+2 \xi_{3} \frac{\partial}{\partial \xi_{2}}+\frac{\partial}{\partial \xi_{3}}\right) \\
		&+(\xi_{1} \xi_{2}-\xi_{1} \xi_{3}^{2}+\xi_{1} \xi_{4}+\xi_{2}^{3}-3 \xi_{2}^{2} \xi_{3}^{2}+3 \xi_{2}^{2} \xi_{4}+3 \xi_{2} \xi_{3}^{4}\\
		&-6 \xi_{2} \xi_{3}^{2} \xi_{4}-\xi_{2} \xi_{3} ) \left(-\frac{\partial}{\partial \xi_{2}}+\frac{\partial}{\partial \xi_{4}}\right)
	\end{aligned}
	\end{equation}
	and $\mathrm{ad}_{Y_{3}} F(0) \not\in D^3$, $(0,0,1)$ must be the only element of $E^2$. Then, let us focus on $E^1$. Since the form of $\mathrm{ad}_{Y_{2}} F$ is so complex that, for the sake of simplicity, only $\mathrm{ad}_{Y_{2}} F(0)$ is shown here 
	\begin{equation}\nonumber
		\mathrm{ad}_{Y^{2}} F(0)=\frac{\partial}{\partial \xi_{2}}.
	\end{equation}
	Noting that $\mathrm{ad}_{Y_{2}} F(0) \in D^2$, it is definite that $(0,1) \not\in E^1$. We also compute the following vector fields at the origin
	\begin{equation}\nonumber
		\mathrm{ad}_{Y^{2}}^{2} F(0)=0,
	\end{equation}
	\begin{equation}\nonumber
		\mathrm{ad}_{Y^{2}}^{3} F(0)=6 \frac{\partial}{\partial \xi_{1}}+6 \frac{\partial}{\partial \xi_{2}}-6 \frac{\partial}{\partial \xi_{4}},
	\end{equation}
	\begin{equation}\nonumber
		\mathrm{ad}_{Y^{1}} \mathrm{ad}_{Y^{2}} F(0)=\frac{\partial}{\partial \xi_{1}}+\frac{\partial}{\partial \xi_{2}}-\frac{\partial}{\partial \xi_{4}}.
	\end{equation}
	Neither $\mathrm{ad}_{Y_{2}}^{3} F(0)$ nor $\mathrm{ad}_{Y_{1}} \mathrm{ad}_{Y_{2}} F(0)$ belongs to $D^2(0)$. Hence, $(0,3), (1,1)\in E^1$. Since $\mathrm{ad}_{Y}^\alpha F(0)=0$ holds for all $\alpha \in {\cal A}_2 \setminus  {\cal G}_{2}(\{(0,3), (1,1)\})$, it is impossible to find any other proper 2-multi-index belonging to $E^1$ yet. This allows us to conclude that $E^1 = \{(0,3), (1,1)\}$.
	Comparing this example with Example \ref{exa_transformation}, the type of \eqref{eq_nse1} determined by using Corollary \ref{col_type} is the same as the type judged from the equivalent lower triangular form of \eqref{eq_nse1}.
\end{example}

\section{Conclusion} 
We have developed a framework to analyze the multi-indices of the functions given by the right-hand sides of the system equations of lower triangular forms. This leads to two classification schemes of lower triangular forms. To expand the application of those two classifications, the problem of whether a nonlinear system is equivalent to a specific type of lower triangular form has also been solved in this paper.

\bibliographystyle{IEEEtran}
\bibliography{root_bib}

\end{document}